\documentclass[a4paper,UKenglish,cleveref,pdftex, thm-restate,numberwithinsect]{lipics-v2019}

\newif\ifreport
\reportfalse

\usepackage{lipsum}

\usepackage{physics}
\usepackage{amssymb}
\usepackage{enumitem}
\usepackage{mathtools}
\usepackage{mathpartir}
\usepackage[utf8]{inputenc}
\usepackage{tikz}
\usetikzlibrary{matrix, arrows.meta, positioning, calc}
\usepackage[usestackEOL]{stackengine}
\usepackage{proof-at-the-end}

\DeclareFontFamily{OT1}{pzc}{}
\DeclareFontShape{OT1}{pzc}{m}{it}{<-> s * [1.200] pzcmi7t}{}
\DeclareMathAlphabet{\mathpzc}{OT1}{pzc}{m}{it}
\renewcommand{\mathscr}{\mathpzc}

\newcommand{\catname}[1]{{\normalfont\textbf{#1}}}

\newcommand{\eqt}[1]{\mathsf{Eq}\qty(\mathcal{#1})}
\newcommand{\mmb}[1]{\mathsf{M}\qty(\mathcal{#1})}

\newcommand{\sgn}{\catname{Sign}}

\newcommand{\supp}[1]{\mathsf{supp}({#1},\mu_{#1})}

\newcommand{\rs}[1]{\mathsf{r}_{\mathbf{#1}}}
\newcommand{\res}[2]{\mathsf{r}_{\mathbf{#1},\hspace{1pt} {#2}}}

\newcommand{\alg}[1]{{#1}\text{-}\mathsf{Alg}}
\newcommand{\arro}[0]{\mathrel{\rotatebox[origin=c]{270}{$\twoheadrightarrow$}}}

\newcommand{\arr}[2]{{#1}{\arro}{\catname{#2}}}
\newcommand{\term}[1]{\mathsf{T}_{#1}}
\newcommand{\trm}[1]{\mathscr{F}_{#1}}
\newcommand{\trma}[1]{\mathscr{F}^{\catname{Set}}_{#1}}
\newcommand{\terms}[1]{\mathcal{{#1}}\text{-}\mathsf{Terms}}
\newcommand{\ed}{\coloneqq}
\newcommand{\thr}[1]{\mathsf{Th}\qty(\mathcal{#1})}
\newcommand{\mode}[1]{\catname{Mod}\qty(#1)}
\newcommand{\modd}[0]{\catname{Mod}}
\newcommand{\id}[1]{\mathrm{id}_{#1}}

\newcommand{\gn}[0]{{\langle \nabla(X)\rangle_{{\mathcal{T}_{\Lambda[\nabla(X)]}},{\bar{\eta}_{\nabla(X)}}}}}
\newcommand{\gen}[3]{\langle {#1}, \mu_{#1} \rangle_{{\mathcal{#3}},{#2}}}
\newcommand{\trms}[1]{\mathsf{Terms}\qty({#1})}
\newcommand{\lan}{\catname{Lng}}
\newcommand{\stat}[1]{\mathsf{Form}\qty(\mathcal{#1})}

\newcommand{\msub}[1]{\mathscr{M}\text{-}\mathscr{Sub}_{\catname{C}}}

\newcommand{\ar}[0]{\mathsf{ar}}

\newcommand{\ex}[2]{\mathsf{E}\qty(#1, #2)}

\newcommand{\ext}[3]{#3^{\mathcal{#1}, {#2}}}

\newcommand{\fuz}[1]{\catname{Fuz}\qty(#1)}

\newcommand*\circled[1]{\tikz[baseline=(char.base)]{
		\node[shape=circle,draw,inner sep=1pt] (char) {#1};}}

\newcommand\functorop[1][l]{\csname#1functor\endcsname}
\newcommand\lfunctorop[3]{%
	\setbox0=\hbox{$#2$}%
	\kern\wd0%
	\ensurestackMath{\Centerstack[c]{#1\\ \mathllap{#2\;\,}\mathclap{\DownArrow}\\#3}}%
}
\newcommand\rfunctorop[3]{%
	\setbox0=\hbox{$#2$}%
	\ensurestackMath{\Centerstack[c]{#1\\\mathclap{\UpArrow}\mathrlap{\,\;#2}\\#3}}%
	\kern\wd0%
}

\setstackgap{L}{1.3\normalbaselineskip}
\newcommand\UpArrow{\rotatebox[origin=c]{90}{$\longrightarrow$\,}}
\newcommand\DownArrow{\rotatebox[origin=c]{-90}{$\longrightarrow$\,}}

\newcommand{\sqt}[1]{\mathsf{Seq}\qty(\mathcal{#1})}
\newcommand\functor[1][l]{\csname#1functor\endcsname}
\newcommand\lfunctor[3]{%
	\setbox0=\hbox{$#2$}%
	\kern\wd0%
	\ensurestackMath{\Centerstack[c]{#1\\ \mathllap{#2\;\,}\mathclap{\DownArrow}\\#3}}%
}
\newcommand\rfunctor[3]{%
	\setbox0=\hbox{$#2$}%
	\ensurestackMath{\Centerstack[c]{#1\\\mathclap{\DownArrow}\mathrlap{\,\;#2}\\#3}}%
	\kern\wd0%
}

\newcommand{\eim}[1]{\catname{Alg}(\mathsf{T}_{#1})}

\setstackgap{L}{1.3\normalbaselineskip}

\newtheorem*{thm1*}{Theorem}

\newtheorem*{lem*}{Lemma}

\newtheorem*{prop*}{Proposition}

\newtheorem*{notazione}{Notation}

\title{Fuzzy algebraic theories}

\author{Davide Castelnovo}
{Department of Mathematics, Computer Science and Physics, University of Udine, Italy}
{davide.castelnovo@uniud.it}{}{}
\author{Marino Miculan}
{Department of Mathematics, Computer Science and Physics, University of Udine, Italy}
{marino.miculan@uniud.it}{0000-0003-0755-3444}{Supported by the Italian MIUR PRIN 2017FTXR7S \emph{IT MATTERS}.}

\authorrunning{D.~Castelnovo, M.~Miculan}
\Copyright{Davide~Castelnovo and Marino Miculan}



\hideLIPIcs
\nolinenumbers 
\ifreport
\nolinenumbers 
\hideLIPIcs  
\else
\relatedversion{A short version of this paper will appear in the \emph{CSL 2022} conference proceedings.} 
\fi

\ifreport
\else
\EventEditors{John Q. Open and Joan R. Access}
\EventNoEds{2}
\EventLongTitle{42nd Conference on Very Important Topics (CVIT 2016)}
\EventShortTitle{CVIT 2016}
\EventAcronym{CVIT}
\EventYear{2016}
\EventDate{December 24--27, 2016}
\EventLocation{Little Whinging, United Kingdom}
\EventLogo{}
\SeriesVolume{42}
\ArticleNo{xx}
\fi

\begin{document}
	
	\maketitle
	
	\begin{abstract}
		In this work we propose a formal system for \emph{fuzzy algebraic reasoning}.
		The sequent calculus we define is based on two kinds of propositions, capturing equality and existence of terms as members of a fuzzy set. 
		We provide a sound semantics for this calculus and show that there is a notion of free model for any theory in this system, allowing us (with some restrictions) to recover models as Eilenberg-Moore algebras for some monad.
		We will also prove a completeness result: a formula is derivable from a given theory if and only if it is satisfied by all models of the theory.
		Finally, leveraging results by Milius and Urbat, we give HSP-like characterizations of subcategories of algebras which are categories of models of particular kinds of theories.
	\end{abstract}

\section{Introduction}	

One of the most fruitful and influential lines of research of Logic in Computer Science is the algebraic study of computation.  
After Moggi's seminal work  \cite{moggi91monads} showed that notions of computation can be represented as monads, Plotkin and Power \cite{plotkinpower02} approached the problem using operations and equations, i.e., Lawvere theories. 
Since then, various extensions of the notion of Lawvere theory have been introduced in order to accommodate an ever increasing number of computational notions within this framework; see, e.g., \cite{plotkinpower03,hylandpower07,nishizawa2009lawvere}, and more recently \cite{bacci2018algebraic,bacci2020quantitative} for quantitative algebraic reasoning for probabilistic computations.


Along this line of research, in this work we study algebraic reasoning on \emph{fuzzy sets}. 
Algebraic structures on fuzzy sets are well known since the seventies (see e.g., \cite{rosenfeld1971fuzzy,mashour1990normal,ajmal1994homomorphism, mordeson2012fuzzy}).
Fuzzy sets are very important in computer science, with applications ranging from pattern recognition to decision making, from system modeling to artificial intelligence.
So, it is natural to ask if it is possible to use an approach similar to above for \emph{fuzzy algebraic reasoning}.

In this paper we answer positively to this question.
We propose a sequent calculus based on two kind of propositions, one expressing equality of terms and the other the existence of a term as a member of a fuzzy set. These sequents have a natural interpretation in categories of fuzzy sets endowed with operations. 
This calculus is sound and complete for such a semantics: a formula is satisfied by all the models of a given theory if and only if it is derivable from it.

It is possible to go further. Both in the classical and in the quantitative settings there is a notion of free model for a theory; we show that is also true for theories in our formal system for fuzzy sets.
In general the category of models of a given theory will not be equivalent to the category of Eilenberg-Moore algebras for the induced monad, but we will show that this equivalence holds for theories with sufficiently simple axioms.
Finally we will use the techniques developed in \cite{milius2019equational} to prove two results analogous to Birkhoff's theorem.

\smallskip 
\noindent\textit{Synopsis.} 
In \cref{sec:fuzzy} we recall the category $\fuz{H}$ of fuzzy sets over a frame $H$. 
\cref{sec:the} introduces the syntax and the rules of fuzzy theories.
Then, in \cref{sec:al} we introduce the notions of algebras for a signature and of models for a theory; in this section we will also show that the calculus proposed is sound and complete. 
\cref{sec:fre} is devoted to free models and it is shown that if a theory is \emph{basic} then its category of models arose as the category of Eilenberg-Moore algebras for a monad on $\fuz{H}$. 
In \cref{sec:birk} we use the results of \cite{milius2019equational} to prove two HSP-like theorems for our calculus. 
Finally, \cref{sec:conc} draws some conclusions and directions for future work. 
\cref{sec:der} contains two detailed derivations used in the paper.

\section{Fuzzy sets}
	\label{sec:fuzzy}
	In this section we will recall the definition and some well-known properties of the category of fuzzy sets over a frame $H$ (i.e. a complete Heyting algebra \cite{johnstone1982stone}).
	\begin{definition}[\cite{wyler1991lecture,wyler1995fuzzy}]
		Let $H$ be a frame. A \emph{$H$-fuzzy set} is a pair $(A, \mu_A)$ consisting in a set $A$ and a \emph{membership function} $\mu_A:A\rightarrow H$. The \emph{support of $\mu_A$} is  the set $	\supp{A}$ of elements $x\in A$ such that $\mu_A(x)\neq \bot$.
		An arrow $f:(A, \mu_A)\rightarrow (B, \mu_B)$ is a function $f:A\rightarrow B$ such that	$\mu_A(x)\leq \mu_B(f(x))$ for all $x\in A$.
		
		We denote by $\fuz{H}$ the category of $H$-fuzzy sets and their arrows. We will often drop the explicit reference to the frame $H$ when there is no danger of confusion.
	\end{definition}

	\begin{proposition}\label{adj}For any frame $H$, the forgetful functor $\mathscr{V}:\fuz{H}\rightarrow \catname{Set}$ has both a left and a right adjoint $\nabla$ and $\Delta$ endowing a set $X$ with the function constantly equal to the bottom and the top element of $H$, respectively.
	\end{proposition}
	\begin{proof}
		If $\nabla(X)$ and $\Delta(X)$ are, respectively $(X,c_\bot)$ and $(X, c_\top)$, where $c_\bot$ and $c_\top$ are the functions $X\rightarrow H$ constant in $\bot$ and $\top$, then for any $X\in \catname{Set}$, $	 	\id{X}:\mathscr{V}(\Delta(X))=X\rightarrow X=\mathscr{V}(\nabla(X))$
		is the counit of  $\mathscr{V}\dashv \Delta$ and the unit of $\nabla \dashv \mathscr{V}$.	
\end{proof}

\begin{definition}
	Let $e:A\rightarrow B$ and $m:C\rightarrow D$ be two arrows in a category $\catname{C}$, we say that $m$ has \emph{the left lifting property with respect to $e$} if for any two arrows $f:A\rightarrow C$ and $g:B\rightarrow D$ such that $m\circ f=g\circ e$ there exists a unique $k:B\rightarrow C$ with $m\circ k= g$.
		
	A \emph{strong monomorphism} is an arrow $m$ which has the left lifting property with respect to all epimorphisms. 
\end{definition}
\begin{proposition}\label{mono}
	Let $f:(A, \mu_A)\rightarrow (B, \mu_B)$ be an arrow of $\fuz{H}$, then:
	\begin{enumerate}
		\item $f$ is a monomorphism iff it is injective; $f$ is an epimorphism iff it is surjective;
		\item $f$ is a strong monomorphism iff it is injective and $\mu_B(f(x))=\mu_A(x)$ for all $x\in A$;
		\item $f$ is a split epimorphism iff for any $b\in B$ there exists $a_b\in f^{-1}(b)$ s.t.~$\mu_B(b)=\mu_A(a_b)$.
	\end{enumerate}
\end{proposition}
	 Let us start with the following observation.
	 \begin{remark}\label{mon}
	 	In any category $\catname{C}$, if $m:A\rightarrow B$ is a strong monomorphism then it is a monomorphism. Suppose that $m\circ f = m\circ g$ for some $f$ and $g$, then we have the square \vspace{-0.5ex}
	 		\begin{center}
	 			\begin{tikzpicture}
	 			\node(A)at(0,0){$C$};
	 			\node(B)at(0,-1.5){$C$};
	 			\node(C)at(1.5,-1.5){$B$};
	 			\node(D)at(1.5,0){$A$};
	 			\draw[->](A)--(B)node[pos=0.5, left]{$\id{C}$};
	 			\draw[->](B)--(C)node[pos=0.5, below]{$m\circ g$};
	 			\draw[->](A)--(D)node[pos=0.5, above]{$f$};
	 			\draw[<-](C)--(D)node[pos=0.5, right]{$m$};
	 			\end{tikzpicture}
	 		\end{center}
 		thus there exists a unique  $k$ such that $m\circ k= m \circ g$ but both $f$ and $g$ satisfy this equation.
 \end{remark}
\begin{remark}
Notice that the previous remark implies that for any square as in the definition of strong monomorphism, $k\circ e=f$.
 \end{remark}

\begin{proof}[Proof of \cref{mono}]
		\begin{enumerate}
			\item 	$\mathscr{V}:\fuz{H}\rightarrow \catname{Set}$ is faithful, so it reflects monomorphisms and epimorphisms, since it is both a left and a right adjoint it preserves them too, hence $1$ follows.
			\item Let us show the two implications.
			\begin{itemize}
				\item[$\Rightarrow$] Injectivity follows from \cref{mon},  let $(D, \mu_D)\ed(f(A), {\mu_B}_{|f(A)})$ be the set-theoretic image of $f$ endowed with the restriction of $\mu_B$. We have a factorization of $f$ as $m\circ e$ where $m$ is a monomorphism and $e$ an epimorphism, so we have the square
				\begin{center}
					\begin{tikzpicture}
					\node(A)at(0,0){$(A, \mu_A)$};
					\node(B)at(0,-1.5){$(D, {\mu_D})$};
					\node(C)at(3,-1.5){$(B,\mu_B)$};
					\node(D)at(3,0){$(A, \mu_A)$};
					\draw[->](A)--(B)node[pos=0.5, left]{$e$};
					\draw[->](B)--(C)node[pos=0.5, below]{$m$};
					\draw[->](A)--(D)node[pos=0.5, above]{$\id{(A,\mu_A)}$};
					\draw[<-](C)--(D)node[pos=0.5, right]{$f$};
					\end{tikzpicture}
				\end{center}
				and its diagonal filling $k:(D, \mu_D)\rightarrow (A, \mu_A)$. For any $x\in A$:
				\begin{align*}
				\mu_B(f(x))=\mu_{B}(m(e(x)))=\mu_D(e(x))\geq  \mu_A(k(e(x)))=\mu_A(x) 
				\end{align*}
				and we get the thesis. 
				\item[$\Leftarrow$]  Any square \vspace{-1ex}
				\begin{center}
					\begin{tikzpicture}
					\node(A)at(0,0){$(C, \mu_C)$};
					\node(B)at(0,-1.5){$(D, {\mu_D})$};
					\node(C)at(2,-1.5){$(B,\mu_B)$};
					\node(D)at(2,0){$(A, \mu_A)$};
					\draw[->](A)--(B)node[pos=0.5, left]{$e$};
					\draw[->](B)--(C)node[pos=0.5, below]{$m$};
					\draw[->](A)--(D)node[pos=0.5, above]{$g$};
					\draw[<-](C)--(D)node[pos=0.5, right]{$f$};
					\end{tikzpicture}
				\end{center}
				induces the square 
				\begin{center}
					\begin{tikzpicture}
					\node(A)at(0,0){$C$};
					\node(B)at(0,-1.5){$D$};
					\node(C)at(1.5,-1.5){$B$};
					\node(D)at(1.5,0){$A$};
					\draw[->](A)--(B)node[pos=0.5, left]{$e$};
					\draw[->](B)--(C)node[pos=0.5, below]{$m$};
					\draw[->](A)--(D)node[pos=0.5, above]{$g$};
					\draw[<-](C)--(D)node[pos=0.5, right]{$f$};
					\end{tikzpicture}
				\end{center}
				in $\catname{Set}$, which, by point $1$ and $2$, has a diagonal filling $k:D\rightarrow A$. Now:
				\begin{align*}
				\mu_D(x)\leq \mu_B(m(x))=\mu_B(f(k(x)))=
				\mu_A(k(x))
				\end{align*} 
				and we can conclude that $k$ is an arrow of $\fuz{H}$.
				\item If $f$ is split then there exists a right inverse $e:(B, \mu_B)\rightarrow (A, \mu_A)$ and we can put $a_b:=e(b)$. The other direction follows noticing that $b\mapsto a_b$ is a splitting of $f$.\qedhere 
			\end{itemize} 
		\end{enumerate}
\end{proof}

\begin{definition}[\cite{kelly1991note}]
	A \emph{proper factorization system} on a category $\catname{C}$ is a pair $(\mathscr{E}, \mathscr{M})$  given by two classes of arrows such that: 
	\begin{itemize}
		\item $\mathscr{E}$ and $\mathscr{M}$ are closed under composition;
		\item every isomorphism belongs to both $\mathscr{E}$ and $\mathscr{M}$;
		\item every $e\in \mathscr{E}$ is an epimorphism  and every $m\in \mathscr{M}$ is a monomorphism;
		\item every $m\in \mathscr{M}$ has the left lifting property with respect to every $e\in \mathscr{E}$;
		\item every arrow of $\catname{C}$ is equal to $m\circ e$ for some $m\in\mathscr{M}$ and $e\in \mathscr{E}$.
	\end{itemize} 
	\end{definition}
	\begin{lemma}\label{prod} For any frame $H$, $\fuz{H}$ has all products. Moreover the classes of epimorphisms and strong monomorphisms form a proper factorization system on it.
	\end{lemma}
	\begin{proof}
$\Delta(1)$ is the terminal object by \cref{adj}. Given a family $\qty{\qty(A_i, \mu_{i})}$ of fuzzy sets, their product is constructed taking their product in $\catname{Set}$ equipped with
\begin{align*}
\mu:\prod_{i\in I}A_i \rightarrow H\qquad
\qty(a_i)_{i\in I} \mapsto \bigwedge_{i\in I}\mu_i(a_i)
\end{align*} The last part of the thesis follows at once  noticing that every arrow $f:(A,\mu_A)\rightarrow (B, \mu_B)$ factors as 
\begin{equation*}
(A, \mu_A )\xrightarrow{e} (f(A), {\mu_{B}}_{|f(A)})\xrightarrow{m}(B, \mu_B)	
\end{equation*}
and applying \cref{mono}.
	\end{proof}	

\begin{remark}
	Completeness and the existence of both adjoints to $\mathscr{V}$ can be deduced directly from the fact that $\fuz{H}$ is topological over $\catname{Set}$ \cite[Prop.~71.3]{wyler1991lecture}.
\end{remark}
	
\section{Fuzzy Theories}\label{sec:the}  
In this section we introduce the syntax and logical rules of fuzzy theories.
The first step is to introduce an appropriate notion of signature. 
Differently from usual signatures, in fuzzy theories constants cannot be seen simply as 0-arity operations, because , as we will see in  \cref{sec:al}, these are interpreted as fuzzy morphisms from the terminal object, and these correspond only to elements whose membership degree is $\top$.

\begin{definition} A \emph{signature} $\Sigma=(
	O, \ar, C)$ is a set $O$ of \emph{operations} with an \emph{arity function} $\ar:O\rightarrow \mathbb{N}_+$ and a set $C$ of \emph{constants}. 
	Signatures form a category $\sgn$ in which an arrow $\Sigma_1=(O_1, \ar_1, C_1)\rightarrow \Sigma_2=(O_2, \ar_2, C_2)$ is a pair $ \mathbf{F}=(F_1, F_2)$ of functions: $F_1:O_1\rightarrow O_2$ and $F_2:C_1\rightarrow C_2$ with the property that $\ar_2\circ F_1= \ar_1$. 
	
	An \emph{algebraic language} $\mathcal{L}$ is a pair $(\Sigma, X)$ where $\Sigma $ is a signature and $X$ a set. The category $\lan$ of algebraic languages is just $\sgn\times \catname{Set}$.
\end{definition}

\begin{example}
	The \emph{signature of semigroups} $\Sigma_S$ in which $O=\{\cdot\}$, $\ar(\cdot)=2$ and $C=\emptyset$. 
\end{example}
\begin{example}
	The \emph{signature of groups} $\Sigma_G$ is equal to $\Sigma_S$ plus an operation $(-)^{-1}$ of arity $1$ and a constant $e$.
\end{example}
Given a language $\mathcal{L}$ we can inductively apply the operations to the set of variables to construct terms, and once terms are built we can introduce equations.
\begin{definition}
	Given a language $\mathcal{L}=(\Sigma, X)$,  the set $\terms{L}$ is the smallest set s.t.
	\begin{itemize}
		\item $X\sqcup C\subset \terms{L}$;
		\item if $f\in O$ and $t_1,\dots, t_{\ar(f)}\in \terms{L}$ then $f\qty(t_1,\dots, t_{\ar(f)})\in \terms{L}$.
	\end{itemize}
\end{definition}
\vspace{-0.5ex}

\begin{proposition}\label{terms}
	There exists a functor $\mathsf{Terms}:\lan \rightarrow \catname{Set}$ sending $\mathcal{L}$ to $\terms{L}$.
\end{proposition}
\begin{proof}
		For any  $F=((F_1,F_2), h):(\Sigma_1, X)\rightarrow (\Sigma_2,Y)$ we can define  $\trms{\mathbf{F}}$ by induction:
	\begin{itemize}
		\item for any $x\in X$, $\trms{F}(x)\ed g(x)$;
		\item for any $c\in C_1$, $\trms{F}(c)\ed F_2(c)$;
		\item for any $f\in O_1$ and $t_1,\dots, t_{\ar_1(f)}\in \terms{L}$ it is $\ar_2(F_1(f))=\ar_1(f)$, so  we can define
		\begin{equation*}
			\trms{F}\qty(f\qty(t_1,\dots,t_{\ar(f)}))\ed F_1(f)\qty(\trms{F}\qty(t_1),\dots,\trms{F}\qty(t_{\ar(f)}) )
		\end{equation*}
	\end{itemize}
	 Identities are preserved and an easy induction shows that composition is respected. 
\end{proof}

\begin{definition}[Formulae]
	For any language $\mathcal{L}$ we define the sets $\eqt{L}$ of \emph{equations} as the product $\eqt{L} \ed \terms{L}\times \terms{L}$ and the set $\mmb{L}$ of \emph{membership propositions} as $\mmb{L}\ed H\times \terms{L}$.
	We will write $s\equiv t$ for $(s,t)\in \eqt{L}$ and  $\ex{l}{t}$ for  $(l,t)\in \mmb{L}$. The set $\stat{L}$ of {formulae} is then defined as $\eqt{L} \sqcup \mmb{L}$. 
\end{definition}	
Clearly, a proposition $s\equiv t$ means ``$s$ and $t$ are equivalent and hence interchangeable''; on the other hand,  $\ex{l}{t}$ intuitively means ``the degree of existence of $t$ is at least $l$''.
\begin{definition}[Sequent ant fuzzy theory]
	A \emph{sequent} $\Gamma \vdash \psi$ is an element $(\Gamma, \psi)$ of $\sqt{L}\ed\mathcal{P}(\stat{L})\times\stat{L}$, where $\mathcal{P}$ is the (covariant) power set functor.	
	A \emph{fuzzy theory in the language $\mathcal{L}$} is a subset $\Lambda\subset\sqt{L}$ and we will use $\thr{L}$ for the set  $\mathcal{P}\qty(\sqt{L})$.
\end{definition}
	
\begin{notazione}\label{not}
	We will write $\vdash \phi$ for $\emptyset \vdash \phi$.
		
	For any function $\sigma: X \rightarrow \terms{L}$ and $t\in \terms{L}$ we denote $t[\sigma]$ the term obtained substituting  $\sigma(x)$ to any occurence of $x$ in $t$.
	Moreover, for any formula $\phi\in \stat{L}$ we define $\phi[\sigma]$ as $t[\sigma]\equiv s[\sigma]$ if $\phi$ is $t\equiv s$ or as $\ex{l}{t[\sigma]}$ if $\phi$ is $\ex{l}{t}$.	Finally, given $\Gamma \subset \mathcal{P}(\stat{L})$ we put $\Gamma[\sigma]\ed\{\phi[\sigma]\mid \phi \in \Gamma\}$.
\end{notazione}

\begin{figure}[t]
\[\begin{array}{ccc}
	\inferrule*[right=A]{ \phi \in \Gamma}{\Gamma \vdash \phi} 	\qquad
	\inferrule*[right=Weak]{\Gamma\vdash \phi}{\Gamma\cup \Delta \vdash \phi} &
	\inferrule*[right=Cut]{\qty{\Gamma\vdash\phi\mid \phi \in \Phi} \\ \Phi \vdash \psi}{\Gamma \vdash \psi} \\
	\inferrule*[right=Refl]{\hspace{1pt}}{\Gamma \vdash s\equiv s} \qquad
	\inferrule*[right=Sym]{\Gamma \vdash s\equiv t}{\Gamma \vdash t\equiv s} &
	\inferrule*[right=Trans]{\Gamma \vdash s \equiv t \\ \Gamma \vdash t \equiv u }{\Gamma \vdash t \equiv u} \\ 
	\inferrule*[right=Sub]{ \sigma:X \rightarrow \terms{L}\\ \Gamma \vdash \psi}{\Gamma[\sigma] \vdash  \psi[\sigma] } &
	\inferrule*[right=Cong]{f\in O\\ \qty{\Gamma \vdash t_i \equiv s_i}_{i=1}^{\ar(f)}}{\Gamma \vdash f\qty(t_1,\dots,t_{\ar(f)}) \equiv f\qty(s_1,\dots,s_{\ar(f)}) } \\
	\inferrule*[right=Inf]{\hspace{1pt}}{\Gamma \vdash \ex{\bot}{t}}\qquad 
	\inferrule*[right=Mon]{\Gamma \vdash  \ex{l}{t}}{\Gamma \vdash  \ex{l\wedge l'}{t}  } &
	\inferrule*[right=Exp]{\qty{\Gamma \vdash \ex{l_i}{t_i}}_{i=1}^{\ar(f)}} {\Gamma \vdash  \ex{\inf\qty(\qty{l_i}_{i=1}^n)}{f\qty(t_1,\dots,t_{\ar(f)}) }}   \\
	\inferrule*[right=Sup]{S\subset H\\\qty{\Gamma \vdash \ex{l}{t}}_{l\in S}}{\Gamma \vdash  \ex{\sup(S)}{t} } &
	\inferrule*[right=Fun]{\Gamma \vdash t \equiv s \\ \Gamma \vdash \ex{l}{t}}{\Gamma \vdash \ex{ l}{s}}
\end{array}\]
\caption{Derivation rules for the fuzzy sequent calculus.}\label{fig:scrules}
\end{figure}	

\begin{definition}
	For any $\mathcal{L}$, the \emph{fuzzy sequent calculus} is given by the rules in \cref{fig:scrules}.
	
	Given a fuzzy theory $\Lambda$,   its \emph{deductive closure} $\Lambda^\vdash$ is the smallest subset of $\sqt{L}$ which contains $\Lambda$ and it is closed under the rules of fuzzy sequent calculus. 
	A sequent is \emph{derivable from $\Lambda$} (or simply derivable if $\Lambda=\emptyset$) if it belongs to ${\Lambda}^{\vdash}$. 
	We will write $\vdash_ {\Lambda}\phi$ if $\vdash \phi \in \Lambda^{\vdash}$. 
	
	Finally we say that two theories $\Lambda$ and $\Theta$ are \emph{deductively equivalent} if $\Lambda^\vdash=\Theta^{\vdash}$.
\end{definition}
	
The next result shows how maps between languages are lifted to theories.
\begin{proposition}
	For any $\mathbf{F}:\mathcal{L}_1\rightarrow \mathcal{L}_2$ arrow of $\lan$:
	\begin{enumerate}
		\item there exists a Galois connection $\mathbf{F}_*\dashv \mathbf{F}^*$ between $\qty(\mathsf{Th}\qty(\mathcal{L}_1), \subset)$ and $\qty(\mathsf{Th}\qty(\mathcal{L}_2), \subset)$;
		\item $	\mathbf{F}_*\qty(\Lambda_1^\vdash)\subset \qty(\mathbf{F}_*\qty(\Lambda_1))^\vdash$ and $\qty(\mathbf{F}^*\qty(\Lambda_2))^\vdash\subset \mathbf{F}^*\qty(\Lambda_2^\vdash)$ for any $\Lambda_1\in \mathbf{Th}(\mathcal{L}_1)$ and $\Lambda_2 \in \mathbf{Th}(\mathcal{L}_2)$.
	\end{enumerate} 
\end{proposition} 
\begin{proof}
	\begin{enumerate}
		\item For any formula $\psi$ we can define (using \cref{terms}):
		\begin{gather*}
			\stat{\mathbf{F}}(t\equiv s)\ed\trms{\mathbf{F}}(t)\equiv\trms{\mathbf{F}}(s)\\\stat{\mathbf{F}}(\ex{l}{t})\ed\ex{l}{\trms{\mathbf{F}}(t)}
		\end{gather*}
		getting a function $	\stat{\mathbf{F}}:\mathsf{Form}\qty(\mathcal{L}_1)\rightarrow \mathsf{Form}\qty(\mathcal{L}_2)$. We can extend it to sequents defining
		\begin{equation*}
			\sqt{\mathbf{F}}(\Phi	\vdash 	\psi)\ed\qty{\stat{\mathbf{F}}(\phi)\mid \phi \in \Phi }\vdash \stat{\mathbf{F}}(\psi)
		\end{equation*}
		Now it is enough to take $\mathbf{F}_*$ and $\mathbf{F}^*$ as the image and the preimage of $\sqt{\mathbf{F}}$ respectively. 
		\item First of all notice that the two inequalities are equivalent because $\mathbf{F}_*\dashv \mathbf{F}^*$. Now, the first one holds since if a sequent $\Gamma \vdash \phi$ follows from the set of sequents $\qty{\Gamma_i\vdash \phi_i}_{i\in I}$ by the application of one of the rules of the fuzzy sequent calculus then the same rule can be applied to $\qty{\sqt{\mathbf{F}}\qty(\Gamma_i\vdash \phi_i)}_{i\in I}$ to get $\sqt{\mathbf{F}}\qty(\Gamma\vdash \psi)$.
			\qedhere
	\end{enumerate} 
\end{proof}

Usually, logics enjoy the so-called ``deduction lemma'', which allows us to treat elements of a theory on a par with assumptions in sequents.  In fuzzy theories, this holds in a slightly restricted form, as proved next.
\begin{lemma}[Partial deduction lemma]\label{ext}
	Let $\Lambda$ be in $\thr{L}$ and $\Gamma\in \mathcal{P}\qty(\stat{L})$, let also $\Lambda[\Gamma]$ be the theory $\Lambda \cup \qty{\emptyset \vdash \phi \mid \phi \in \Gamma}$. Then the following are true:
	\begin{enumerate}
		\item $\Gamma \cup \Delta \vdash \psi$ in $\Lambda^{\vdash}$  implies $\Delta \vdash \psi$ in $\qty(\Lambda[\Gamma])^{\vdash}$;
		\item if $\Delta \vdash \psi$ is derivable from $\Lambda[\Gamma]$ without using rule \textsc{Sub} then $\Gamma\cup \Delta \vdash \psi $ is in $\Lambda^{\vdash}$.\qedhere
	\end{enumerate} 
\end{lemma}
\begin{proof}
	\begin{enumerate}
		\item By hypothesis $\Gamma \cup \Delta \vdash \psi$ is in $\Lambda^\vdash$ then, since $\Lambda \subset \Lambda[\Gamma]$, it is also in $\qty(\Lambda[\Gamma])^\vdash$.  Now, for any $\phi \in \Gamma$ and $\theta \in \Delta$ rules \textsc{Weak} and A give
		\begin{equation*}
			\inferrule*[right=Weak]{\vdash \phi }{\Delta \vdash \phi}\qquad 
			\inferrule*[right=A]{\hspace{1pt} }{\Delta \vdash \theta}
		\end{equation*}
		so $\qty{\Delta \vdash \phi \mid \phi \in \Gamma \cup \Delta }$ is contained in $ \qty(\Lambda[\Gamma])^{\vdash}$ and we can apply  \textsc{Cut} to get the thesis: 
		\begin{equation*}
			\inferrule*[right=Cut]{\qty{\Delta \vdash \phi \mid \phi \in \Gamma \cup \Delta }\\ \Gamma \cup \Delta \vdash \psi }{\Delta \vdash \psi}
		\end{equation*}
		\item Let us proceed by induction on a derivation of $\Delta \vdash \psi$ from $\Lambda[\Gamma]$. 
		\begin{itemize}
			\item If $\Delta \vdash \psi\in \Lambda[\Gamma]$ then  or $\Delta \vdash \psi\in \Lambda$ and we are done, or $\psi\in \Gamma$  and we can conclude since $\Gamma \cup \Delta \vdash \phi_i$ is in $\Lambda^\vdash$ by rules A and \textsc{Weak}
			\item If $\Delta\vdash \psi$ follows from the application of one of the rules A, \textsc{Inf} or \textsc{Refl} then it belongs to the closure of any theory, by \textsc{Weak} this is true even for $\Gamma \cup \Delta\vdash \psi $ which, in particular, it belongs to $\Lambda^\vdash$.
			\item Suppose that $\Delta \vdash \psi$ comes from an application of \textsc{Weak}, then there exists $\Psi$ and $\Phi$ such that $\Psi \cup \Phi=\Delta$ and $\Psi\vdash \phi$ is in $(\Lambda[\Gamma])^\vdash$. By inductive hypothesis we have that
			\begin{equation*}
				\inferrule*[right=Weak]{\Gamma  \cup \Psi  \vdash \psi}{\Gamma  \cup \Psi \cup \Phi \vdash \psi  }
			\end{equation*}
			is a derivation of $\Gamma \cup \Delta \vdash \psi $ from $\Lambda$.
			\item If $\Delta \vdash \psi$ is derived with an application of \textsc{Cut} as last rule then there exists a set $\Theta$ such that $
				\qty{\Delta \vdash \theta \mid \theta \in \Theta} \cup \qty{\Theta\vdash \psi} \subseteq \qty(\Lambda[\Gamma])^\vdash
			$, therefore, by the inductive hypothesis, we have that $\qty{\Gamma \cup \Delta \vdash \theta \mid \theta \in \Theta} \cup \qty{\Gamma \cup \Theta\vdash \psi}$ is conteined in  $\Lambda^\vdash$.
			Now, by rule A we have that $\qty{\Gamma \cup \Delta \vdash \phi \mid \phi \in \Gamma\cup \Theta}$ is contained in $\Lambda^\vdash$ too and so we can apply \textsc{Cut}:
			\begin{equation*}
				\inferrule*[right=Cut]{\qty{\Gamma \cup \Delta\vdash \phi \mid \phi \in \Gamma\cup \Theta}\\ \Gamma \cup \Theta \vdash \psi}{\Gamma \cup \Delta \vdash \psi}
			\end{equation*}
			which gives us the thesis.
			\item Any other rule is of the form
			\begin{equation*}
				\inferrule*[right=R]{\qty{\Psi \vdash \xi_j}_{j\in J}}{\Psi \vdash \xi}
			\end{equation*}
			therefore, if $\Delta \vdash \psi$ is derived with an application of one of this rules then the set of its premises must be an element of $(\Lambda[\Gamma])^\vdash$ of type $\qty{\Delta \vdash \xi_j}_{j\in J}$, so by inductive hypothesis $\qty{\Gamma\cup \Delta \vdash \theta_j}_{j\in J}\subset \Lambda^\vdash$ and then
			\begin{equation*}
				\inferrule*[right=R]{\qty{\Gamma \cup \Delta \vdash \xi_j}_{j\in J}}{\Gamma \cup \Delta \vdash \psi}
			\end{equation*}
			witnesses the thesis.\qedhere
		\end{itemize}
	\end{enumerate} 
\end{proof}

\begin{example}\label{ex:semig1}
	Our first set of running examples is inspired by \cite{mordeson2012fuzzy}.
	Let $\Sigma_{S}$ be the signature of semigroups and $X$ a countable set. The theory of \emph{fuzzy semigroups} $\Lambda_S$ is simply the usual theory of semigroups, i.e given by the sequent (using infix notation) 
	\begin{gather*}
		\vdash (x\cdot y)\cdot z \equiv x\cdot (y \cdot z)
	\end{gather*}
	We get the theory of \emph{left ideal} $\Lambda_{LI}$ if we add the axioms (one for any $l\in L$):
	\begin{equation*}
		\ex{l}{y}	\vdash \ex{l}{x\cdot y}
	\end{equation*} 
	Similarly the theory $\Lambda_{RI}$ of \emph{right ideal} is obtained from the axioms:
	\begin{equation*}
		\ex{l}{x}	\vdash \ex{l}{x\cdot y}
	\end{equation*} 
	Finally we get the theory of \emph{(bilateral) ideal} $\Lambda_I$ taking the union of the above theories.
\end{example}

\begin{example}[\cite{rosenfeld1971fuzzy, ajmal1994homomorphism, ajmal1992fuzzy}]\label{ex:fuzzyg1}
	Let $\Sigma_{G}$ be the signature of groups and $X$ a countable set.	The theory $\Lambda_G$ of \emph{fuzzy groups} is simply the usual theory of groups, i.e that given by the sequents
	\begin{gather*}
		\vdash x \cdot x^{-1} \equiv e  \quad
		\vdash  x^{-1} \cdot x \equiv e  \quad
		\vdash e\cdot x \equiv x  \quad
		\vdash x \cdot x \equiv x   \quad
		\vdash (x\cdot y)\cdot z\equiv x\cdot (y \cdot z)
	\end{gather*}
	We get the theory $\Lambda_N$ of \emph{normal fuzzy groups} (\cite{mashour1990normal}) if we add the axioms:
	\begin{equation*}
		\ex{l}{x}	\vdash \ex{l}{y\cdot (x \cdot y^{-1})}
	\end{equation*}
\end{example}
\vspace{-1ex}
\begin{proposition}\label{equ} 
	The sequent
$\ex{l}{x} \vdash \ex{l}{e}$   is derivable from $\Lambda_G$ and $	\ex{l}{y\cdot (x \cdot y^{-1})}\vdash \ex{l}{x}$ is derivable from $\Lambda_N$.
\end{proposition}
\begin{proof}
		The two derivations are shown in \cref{sec:der}. In the second derivation we have used the fact that $\vdash y \equiv \qty(y^{-1})^{-1}$ and  $\vdash x\equiv y \cdot\qty (\qty(y^{-1} \cdot (x\cdot y)) \cdot y^{-1})$ are derivable in the theory of groups (this can be shown as in the case of ordinary group theory) and we have substituted $y^{-1}\cdot (x \cdot y)$ for $x$ in the first application of \textsc{Sub} and $y^{-1}$ for $y$ in the second one. 
\end{proof}
	
\section{Fuzzy algebras and semantics}\label{sec:al}
In this section we provide a sound and complete semantics to the syntax and sequents introduced in \cref{sec:the}. The first step is to define the semantic counterpart of a signature.
\begin{definition}
	Given a signature $\Sigma$, a $\Sigma$-\emph{fuzzy algebra} $\mathcal{A}\ed((A, \mu_A), \Sigma^{\mathcal{A}})$ is a fuzzy set $(A, \mu_A)$ and a collection $\Sigma^{\mathcal{A}}=\{f^{\mathcal{A}}\mid f\in O\}\sqcup \{c^{\mathcal{A}}\mid c\in C\}$ of arrows:
\begin{align*}
	f^{\mathcal{A}}:(A,\mu_A)^{\ar(f)}\rightarrow (A, \mu_A)\qquad
	c^{\mathcal{A}}:(1, c_\bot)\rightarrow (A, \mu_A)
\end{align*}
where $c_\bot$ is the constant function in $\bot$.
A \emph{morphism of $\Sigma$-fuzzy algebras} $\mathcal{A} \rightarrow \mathcal{B}$ is an arrow $g:(A, \mu_A)\rightarrow (B, \mu_B)$ such that $g\circ c^{\mathcal{A}}=c^{\mathcal{B}}$ and  $f^{\mathcal{B}}\circ g^{\ar(f)}=g\circ f^{\mathcal{A}}$ for every $c\in C$ and $f\in O$.	 We will write $\alg{\Sigma}$ for the resulting category of $\Sigma$-fuzzy algebras.
\end{definition}
\begin{remark}
	We will not distinguish between a function from the singleton and its value.
\end{remark}

\begin{definition}
	Let $\mathcal{L}=(\Sigma, X)$ be a language and $\mathcal{A}=\qty((A, \mu_A), \Sigma^{\mathcal{A}})$ be a $\Sigma$-algebra.
	
	An \emph{assignment} is simply a function $\iota:X\rightarrow A$. We define the \emph{evaluation in $\mathcal{A}$ with respect to $\iota$} as the function $\ext{A}{\iota}{(\text{-})}:\terms{L}\rightarrow A$ by induction:
	\begin{itemize}
		\item $\ext{A}{\iota}{x}\ed\iota(x)$ if $x\in X$;
		\item $\ext{A}{\iota}{c}\ed c^{\mathcal{A}}$ if $c\in C$;
		\item $\ext{A}{\iota}{(f(t_1,\dots,t_{\ar(f)}))}\ed f^{\mathcal{A}}(\ext{A}{\iota}{t_1}, \dots, \ext{A}{\iota}{t_{\ar(f)}})$ if $f\in O$.
	\end{itemize}
\end{definition}

\begin{proposition}\label{su}
	Let $\mathcal{A}$ be a $\Sigma$-algebra. Given a function $\sigma :X\rightarrow \terms{L}$ and an assignment $\iota:X \rightarrow A$ define $
	\iota_\sigma:X\rightarrow A$ as the assignment sending		$x$ to $ \qty(\sigma\qty(x))^{\mathcal{A}, \iota}$.
	Then $\mathcal{A}\vDash_\iota \phi[\sigma]$ if and only if $\mathcal{A}\vDash_{\iota_\sigma} \phi$. 
\end{proposition}
\begin{proof}
		This follows at once noticing that $\big{(}t\qty[\sigma]\big{)}^{\mathcal{A}, \iota}=t^{\mathcal{A}, \iota_\sigma}$ holds for every term $t$.	
\end{proof}

\begin{definition}
	$\mathcal{A}$ \emph{satisfies $\phi\in \stat{L}$ with respect to $\iota$}, and we write $\mathcal{A}\vDash_{\iota} \phi$, if $\phi$ is $\ex{l}{ t}$ and $l \leq \mu_A(\ext{A}{\iota}{t})$ or if $\phi$ is $t \equiv s$ and $ \ext{A}{\iota}{t}=\ext{A}{\iota}{s}$.
	
	$\mathcal{A}$ \emph{satisfies $\phi$} if $\mathcal{A}\vDash_{\iota} \phi $ for all $\iota:X\rightarrow A$, and we write $\mathcal{A}\vDash \phi$.
	
	Given $\Gamma\subset \stat{L}$, $\mathcal{A}\vDash \Gamma$ ($\mathcal{A}\vDash_{\iota} \Gamma$) means $\mathcal{A}\vDash \phi$ ($\mathcal{A}\vDash_{\iota} \phi$) for any $\phi \in \Gamma$.
	Finally, given a sequent $\Gamma \vdash \phi$ we say that $\mathcal{A}$ \emph{satisfies it with respect to $\iota$} and we will write $\Gamma\vDash_{\mathcal{A}, \iota} \phi$ if $\mathcal{A}\vDash_{\iota} \phi$ whenever $\mathcal{A}\vDash_{\iota}\Gamma$; if this happens for all assignments  $\iota$ we say that $\mathcal{A}$ \emph{satisfies} the sequent and we will write $\Gamma\vDash_{\mathcal{A}} \phi$. 
	
	We say that a $\Sigma$-fuzzy algebra $\mathcal{A}$ is a \emph{model} of a fuzzy theory $\Lambda\in \thr{L}$ if it satisfies all the sequents in it. 
	
	We will write $\mode{\Lambda}$ for the full subcategory of $\alg{\Sigma}$ given by the models of $\Lambda$.
\end{definition}
	Clearly $\alg{\Sigma}=\mode{\emptyset}$. For any $\Lambda\in \thr{L}$ there exist two forgetful functors $\mathscr{U}_\Lambda:\mode{\Lambda}\rightarrow \fuz{L}$ and $\mathscr{V}_\Lambda:\mode{\Lambda}\rightarrow \catname{Set}$.
	We will write $\mathscr{U}_\Sigma$ and $\mathscr{V}_\Sigma$  for $\mathscr{U}_\emptyset$ and $\mathscr{V}_\emptyset$.
		
	\begin{proposition}\label{left}
		For any signature $\Sigma$, $\mathscr{V}_\Sigma$ has a left adjoint $\trma{\Sigma}:\catname{Set}\rightarrow \mode{\Lambda}$.
	\end{proposition}
	\begin{proof}For any set $X$ take the language $\mathcal{L}_X$ and define  $\trma{\Sigma}(X)$ has $(\nabla(\mathcal{L}_X\text{-}\mathbf{Terms}), \Sigma^{\trma{\Sigma}(X)})$ where $c^{\trma{\Sigma}(X)}:=c$ and
		\begin{align*}
			f^{\trma{\Sigma}(X)}: \nabla(\mathcal{L}_X\text{-}\mathbf{Terms})^{\ar(f)}\rightarrow \nabla(\mathcal{L}_X\text{-}\mathbf{Terms})\qquad
			\qty(t_1,\dots, t_{\ar(f)})\mapsto f(t_1,\dots, t_{\ar(f)})
		\end{align*}
		for any $f\in O$. Now,  it is immediate to see that for any $\iota:X\rightarrow \mathscr{V}_\Sigma(\mathcal{A})$ the evaluation $\ext{A}{\iota}{(\text{-})}$ is the unique morphism of $\alg{\Sigma}$ that composed with the inclusion $X\rightarrow\mathcal{L}_X\text{-}\mathbf{Terms}$ gives back $\iota.$ 	
\end{proof}

We now provide two technical results about interpretations.
The first describes how interpretations are moved along morphisms of algebras.
\begin{proposition}\label{ind}
	Let $\mathcal{L}=(\Sigma, X)$ be a language, $\Lambda \in \thr{L}$ and $\mathcal{A}=\qty(\qty(A, \mu_A), \Sigma^{\mathcal{A}}) $, $\mathcal{B}=\qty(\qty(B, \mu_B), \Sigma^{\mathcal{B}})$ be two $\Sigma$-algebras. 
	Let also $f:\mathcal{A}\rightarrow \mathcal{B}$ be a morphism between them, then:
	\begin{enumerate}
		\item $\mathcal{A}$ is a model of $\Lambda$ if and only if it is a model of $\Lambda^{\vdash}$;
		\item $f\circ (-)^{\mathcal{A}, \iota}= (-)^{\mathcal{B}, f\circ \iota}$ for every assignment $\iota:X\rightarrow A$;
		\item for any assignment $\iota:X\rightarrow A$, $\mathcal{A}\vDash_\iota \phi$ entails $\mathcal{B}\vDash_{f\circ\iota} \phi$;
		\item if $\mathscr{U}_\Sigma(f)$ is a strong monomorphism in $\fuz{H}$ and $\iota:X\rightarrow A$ is an assignment then, for any formula $\phi$, $\mathcal{A}\vDash_\iota \phi$ if and only if $\mathcal{B}\vDash_{f\circ \iota} \phi$;
		\item if $\mathscr{U}_\Sigma(f)$ is a strong monomorphism in $\fuz{H}$ and $\mathcal{B}\in \mode{\Lambda}$ then  $\mathcal{A}\in \mode{\Lambda}$.  
	\end{enumerate}
\end{proposition}
\begin{proof}
		\begin{enumerate}
			\item The implication from the right to the left is obvious since $\Lambda \subseteq \Lambda^{\vdash}$. The other follows from soundness.
			\item This follows at once by structural induction. 
			\item This follows from the previous point: if $\phi$ is $t\equiv s$ then
			\begin{align*}
				t^{\mathcal{B}, f\circ \iota}=f\qty(t^{\mathcal{A},  \iota})=f\qty(s^{\mathcal{A},  \iota})=s^{\mathcal{B}, f\circ \iota}
			\end{align*}  
			otherwise
			\begin{align*}
				l \leq \mu_A\qty(t^{\mathcal{A},  \iota})\leq \mu_B\qty(f\qty(t^{\mathcal{A},  \iota}))= 	\mu_B\qty(t^{\mathcal{B}, f\circ \iota})
			\end{align*}
			\item We have to show only the implication from the right to the left. Let us proceed by cases remembering that by \cref{mono} $\mathscr{U}_\Sigma(f)$ is a strong monomorphism if and only if it is injective and $\mu_A(a)=\mu_B(f(a))$ for any $a \in A$.
			\begin{itemize}
				\item $\phi=t\equiv s$.  $\mathcal{A}\vDash_{f\circ \iota} t\equiv s$ if and only if  $t^{\mathcal{A}, \iota}=s^{\mathcal{A},  \iota}$, by injectivity of $f$ this is equivalent to 	$f(t^{\mathcal{A, \iota}})=f(s^{\mathcal{A, \iota}})$ and by point $2$ this is equivalent to $t^{\mathcal{B},f\circ  \iota}=s^{\mathcal{B},f\circ \iota}$, i.e. $\mathcal{B}\vDash_{f\circ \iota}\phi$. 
				\item $\phi=\ex{l}{t}$. As before, $\mathcal{A}\vDash_{ \iota} \phi$ if and only if $\mu_A(t^{\mathcal{A},  \iota})\geq l$ but
				\begin{align*}
					\mu_A(t^{\mathcal{A},  \iota})=	\mu_B(f(t^{\mathcal{A},  \iota})) 	=\mu_B(t^{\mathcal{B}, f\circ \iota})	
				\end{align*}  
				and thus  $\mathcal{A}\vDash_{ \iota} \phi$ is equivalent to $\mathcal{B}\vDash_{f\circ \iota} \phi$.
			\end{itemize} 
			\item Let $\Gamma \vdash \phi $ be in $\Lambda$ and suppose $\mathcal{A}\vDash_\iota \Gamma$ for some assignment $\iota:X\rightarrow B$, by the previous point this implies $\mathcal{B}\vDash_{f\circ \iota }\Gamma$, so, by hypothesis $\mathcal{B}\vDash_{f\circ \iota }\phi$ and we can conclude again using point $3$.\qedhere 
		\end{enumerate} 
\end{proof}

We can also move interpretations and theories along morphisms of signatures.
\begin{definition}
	For any  $\mathbf{F}:\Sigma_1\rightarrow \Sigma_2$ arrow of $\sgn$ and any $\mathcal{A}=\qty(\qty(A, \mu_A), \Sigma_2^{\mathcal{A}})\in \alg{\Sigma_2}$, we define $\rs{F}(\mathcal{A})=\qty(\qty(A,\mu_A), \Sigma_1^{\rs{F}(\mathcal{A})})\in \alg{\Sigma_1}$ putting, for any $f\in O_1$ 
	\begin{equation*}
			f^{\rs{F}(\mathcal{A})}:(A, \mu_A)^{\ar(f)}\rightarrow (A, \mu_A)\qquad
			\qty(a_1,\dots, a_{\ar(f)})\mapsto F_2(f)^{\mathcal{A}}\qty(a_1,\dots, a_{\ar(f)})
	\end{equation*}
	and 	$c^{\rs{F}\qty(\mathcal{A})}:= F_3(c)^{\mathcal{A}}$ for every $c\in C_1$.
\end{definition}\begin{lemma}\label{mrp}
	Let $\mathcal{L}_1=(\Sigma_1, X)$ and $\mathcal{L}_2=(\Sigma_2, Y)$ and $\mathbf{F}=\qty(\qty(F_1,F_2), g):\mathcal{L}_1\rightarrow \mathcal{L}_2$, then:
	\begin{enumerate}
		\item there exists a functor $\rs{F}:\alg{\Sigma_2}\rightarrow \alg{\Sigma_1}$ sending $\mathcal{A}$ to $\rs{F}(\mathcal{A})$; 
		\item  $t^{\rs{F}(\mathcal{A}), \iota \circ g}=\qty(\trms{\mathbf{F}}(t))^{\mathcal{A}, \iota}$ for any assignment $\iota:Y\rightarrow A$ and $t\in \mathcal{L}_1\text{-}\mathsf{Terms}$;
		\item for any assignment $\iota:Y\rightarrow A$, 	$\rs{F}(\mathcal{A})\vDash_{ \iota \circ g } \phi$ if and only if $\mathcal{A}\vDash_\iota \stat{\mathbf{F}}(\phi)$;
		\item  If $X = Y$ and $g=\id{X}$ then $\rs{F}$ restricts to a functor $\res{F}{\Lambda}:\mode{\Lambda}\rightarrow \mode{\mathbf{F}^*(\Lambda)}$.
	\end{enumerate}
\end{lemma}
\begin{proof}
		\begin{enumerate}
			\item Let $h:\mathcal{A}\rightarrow \mathcal{B}$ be a morphism of $\alg{\Sigma_2}$ then $h\circ u^{\mathcal{A}}=u^{\mathcal{B}}\circ g^{\ar_2(h)}$ and $ d^{\mathcal{B}}=h\circ d^{\mathcal{A}}$ for all $u\in O_2$ and $d\in C_2$. In particular this is true when $u=F_2(f)$ and $d=F_3(c)$ for $f\in O_1$ and $c\in C_1$ but these are exactly the conditions making $h$ into an homomorphism $\rs{F}(\mathcal{A})\rightarrow \rs{F}(\mathcal{B})$.
			\item  This follows by induction:
			\begin{itemize}
				\item for any $x\in X$
				\begin{align*}
					x^{ \rs{F}(\mathcal{A}),\iota \circ g}=\iota(g(x))=(\trms{\mathbf{F}}(x))^{\mathcal{A}, \iota}
				\end{align*}
				and for any constant $c\in C_1$
				\begin{align*}
					c^{ \rs{F}(\mathcal{A}),\iota \circ g}=c^{ \rs{F}(\mathcal{A})}=F_3(c)^{\mathcal{A}}=(\trms{\mathbf{F}}(c))^{\mathcal{A}}=
					(\trms{\mathbf{F}}(c))^{\mathcal{A},\iota}
				\end{align*}
				\item for any $f\in O_1$:
				\begin{align*}
					\qty(f\qty(t_1,\dots,t_{\ar_1(f)}))^{\mathcal{\rs{F}(A)}, \iota \circ g }&=f^{\rs{F}(\mathcal{A})}\qty(t^{\mathcal{\rs{F}(A)}, \iota \circ g }_1,\dots,t^{\mathcal{\rs{F}(A)}, \iota \circ g }_{\ar_1(f)})\\&=f^{\rs{F}(\mathcal{A})}\qty(\qty(\trms{\mathbf{F}}(t_1))^{\mathcal{A}, \iota},\dots, \qty(\trms{\mathbf{F}}(t_{\ar_1(f)}))^{\mathcal{A}, \iota})\\&=F_1(f)^{\mathcal{A}}\qty(\qty(\trms{\mathbf{F}}\qty(t_1))^{\mathcal{A}, \iota},\dots, \qty(\trms{\mathbf{F}}\qty(t_{\ar_1(f)}))^{\mathcal{A}, \iota})\\&=
					\qty(F_1(f)\qty(\qty(\trms{\mathbf{F}}\qty(t_1)),\dots, \qty(\trms{\mathbf{F}}\qty(t_{\ar_1(f)}))))^{\mathcal{A}, \iota}\\&=\trms{\mathbf{F}}	\qty(f\qty(t_1,\dots,t_{\ar_1(f)}))^{\mathcal{A}, \iota}
				\end{align*}
			\end{itemize}
			\item This follows immediately by point $2$:
			\begin{align*}
				\rs{F}(\mathcal{A})\vDash_{\iota \circ g}(t\equiv s)&\iff t^{\rs{F}(\mathcal{A}), \iota \circ g }= s^{\rs{F}(\mathcal{A}), \iota \circ g}\\&\iff \qty(\trms{\mathbf{F}}\qty(t))^{\mathcal{A},\iota }=\qty(\trms{\mathbf{F}}\qty(s))^{\mathcal{A}, \iota}\\&\iff \mathcal{A}\vDash_\iota \stat{\mathbf{F}}(t\equiv s)\\\hspace{1pt}\\
				\rs{F}(\mathcal{A})\vDash_{\iota \circ g }\ex{l}{t}&\iff l\leq \mu_A\qty(t^{\rs{F}(\mathcal{A}), \iota \circ g})\\&\iff l \leq  \mu_A\qty(\trms{\mathbf{F}}\qty(t))^{\mathcal{A}, \iota}\\&\iff \mathcal{A}\vDash_\iota \stat{\mathbf{F}}\qty(\ex{l}{t})
			\end{align*}
			\item  Let $\Gamma \vdash \psi$ be in $\mathbf{F}^*(\Lambda)$, by the previous point, for any assignment $\iota:X\rightarrow A$ 
			\begin{equation*}
				\rs{F}(\mathcal{A})\vDash_{ \iota\circ g} \Gamma\iff \mathcal{A}\vDash_\iota \qty{\stat{\mathbf{F}}(\phi)}_{\phi \in \Gamma}
			\end{equation*} 
			Since $\qty{\stat{\mathbf{F}}(\phi)}_{\phi \in \Gamma}\vdash \stat{\mathbf{F}}(\psi)$ is in $\Lambda$ we have  $\mathcal{A}\vDash_\iota \stat{\mathbf{F}}\qty(\psi)$ but, using again point $3$, this implies $ \rs{F}(\mathcal{A})\vDash_\iota \phi$. \qedhere
		\end{enumerate} 
\end{proof}

\begin{example}\label{ex:semig2}
	The models for $\Lambda_S$, $\Lambda_ {LI}$, $\Lambda_{RI}$ and $\Lambda_I$ (\cref{ex:semig1}) are precisely the structures defined in \cite{mordeson2012fuzzy}, while
	the models for $\Lambda_G$ (\cref{ex:fuzzyg1}) are precisely the fuzzy groups as in \cite{rosenfeld1971fuzzy} and those of $\Lambda_N$ are the structures called \emph{normal fuzzy subgroups} in \cite{ajmal1992fuzzy, ajmal1994homomorphism, mashour1990normal}.
\end{example}

\subparagraph{Soundness} 
Now we can proceed proving the soundness of the rules in \cref{fig:scrules}.
\begin{lemma}\label{rules} 
	Let  $\mathcal{L}=(\Sigma, X)$ be a language and $\mathcal{A} =((A, \mu_A), \Sigma^{\mathcal{A}})$ a $\Sigma$-algebra, then:
	\begin{enumerate}
		\item for any assignment $\iota:X\rightarrow A$ and rule
		\vspace{-1ex}
		\begin{equation*}
			\inferrule*[right=R]{\qty{\Psi_i \vdash \xi_i}_{i\in I}}{\Psi \vdash \xi}
		\end{equation*}
		different from \textsc{Sub}, if $\Psi_i\vDash_{\mathcal{A}, \iota} \xi_i$ for all $i\in I$ then  $\Psi \vDash_{\mathcal{A}, \iota} \xi $ too;
		\item for any $\sigma:X\rightarrow \terms{L}$, if $\Gamma\vDash_{\mathcal{A}}\psi$ then $\Gamma[\sigma]\vDash_{\mathcal{A}}\psi[\sigma]$.
	\end{enumerate}	
\end{lemma}
\begin{proof}
	\begin{enumerate}
		\item The thesis is vacuous for rule A and it is trivial for \textsc{Refl}, \textsc{Sym} and \textsc{Trans}. For rule  \textsc{Weak} if  $\mathcal{A}\vDash_{\iota}\Gamma\cup \Delta$ then $\mathcal{A}\vDash_{\iota}\Gamma$ and so, since $\Gamma \vDash_{\mathcal{A}, \iota}\phi$ this implies $\mathcal{A}\vDash_{\iota} \phi$. The thesis for \textsc{Inf} it follows from the fact that $\bot$ is the bottom element, for \textsc{Cong} and \textsc{Fun} it follows since $\mu_A$ and $f^{\mathcal{A}}\in \Sigma^{\mathcal{A}}$   are functions for any $f\in O$. We're left with four rules.
		\begin{itemize}\item[] \textsc{Cut}. For any assignment $\iota$, if $\mathcal{A}\vDash_\iota \Gamma $ then, since  $\Gamma \vDash_{\mathcal{A}, \iota} \phi_i$ for any $i\in I$ we have $\mathcal{A}\vDash_\iota \qty{\phi_i}_{i\in I} $, but $\qty{\phi_i}_{i\in I}\vDash_{\mathcal{A}, \iota} \psi$ and this implies $\mathcal{A}\vDash_\iota \psi$.
			\item[]\textsc{Sup}. Let $\iota:X\rightarrow A$ be such that $\mathcal{A}\vDash_{\iota}\Gamma$, then $\Gamma\vDash_{\mathcal{A}, \iota}\ex{l}{t}$ for all $l\in \mathcal{L}$ implies that $\mu_A(t^{\mathcal{A}, \iota})$ is an upper bound of $\mathcal{L}$, so the thesis follows.
			\item[]\textsc{Mon}. As before let $\iota$ be an assignment for which $\mathcal{A}$ satisfies all the formulae in $\Gamma$, then, since $\Gamma\vDash_{\mathcal{A}, \iota} \ex{l}{t}$ we have that $l\leq \mu_A(t^{\mathcal{A}, \iota})$, so, for any $l'\in H$, $l'\wedge l\leq \mu_A(t^{\mathcal{A}, \iota})$ holds too.
			\item[]\textsc{Exp}. For any $\iota$  such that $ \mathcal{A} \vDash_{ \iota}\Gamma$, $\Gamma\vDash_{\mathcal{A}, \iota}\ex{l_i}{t_i}$ entails  $l_i\leq \mu_A(t^{\mathcal{A}, \iota}_i)$ for any $1\leq i\leq \ar(f)$, therefore:
			\begin{align*}
				\bigwedge_{i=1}^{\ar(f)} l_i &\leq \bigwedge_{i=1}^{\ar(f)}\mu_A(t^{\mathcal{A}, \iota}_i)=
					\mu_{A^n}(t^{\mathcal{A}, \iota}_1,\dots,t^{\mathcal{A}, \iota}_{\ar(f)})\leq \mu_A(f^{\mathcal{A}}(t^{\mathcal{A}, \iota}_1,\dots,t^{\mathcal{A}, \iota}_{\ar(f)}))			 
			\end{align*}
			and we are done.
		\end{itemize}
			
		\item Let $\iota:X \rightarrow A$ an assignment such that $\mathcal{A}\vDash_{\iota}\phi_i[\sigma]$ for any $\phi_i\in \Gamma$. By  \cref{su} this means that $\mathcal{A}\vDash_{\iota_\sigma}\Gamma$ and so, by hypothesis, $\mathcal{A}\vDash_{\iota_\sigma}\psi$ which, again by  \cref{su}, entails $\mathcal{A}\vDash_{\iota}\psi[\sigma]$. 
		\qedhere 
	\end{enumerate} 
\end{proof}
	
\begin{corollary}[Soundness]\label{sound}
	If a $\Sigma$-algebra satisfies all the premises of a rule of the fuzzy sequent calculus then it satisfies also its conclusion.  
\end{corollary}
	
\begin{remark}\label{rem:nodeduclemma}
	Let us see why the deduction lemma (\cref{ext}) cannot be extended to rule \textsc{Sub}. Take $\Sigma$ to be the empty set,  $X=\{x,y, z\}$ and $H=\{0,1\}$. Notice that $\alg{\Sigma}=\fuz{H}$. We have the derivation
	\vspace{-1ex}
	\begin{equation*}
		\inferrule*[Right=Sub]{\vdash x\equiv y}{\vdash x\equiv z}
	\end{equation*}
	If the deduction lemma held for \textsc{Sub}, $x\equiv y\vdash x\equiv z$ would be in $\emptyset^{\vdash}$, hence satisfied by any fuzzy set, but $(H, \id{H})$ with $\iota:	X\rightarrow H$ sending  $x$ and $y$ to $0$ and $z$ to $1$ is a counterexample. 
\end{remark}
\begin{remark}
	Let us take $\Sigma=\emptyset$ and $H=\{0,1\}$ and $X=\{x,y,z\}$ and the derivation as in \cref{rem:nodeduclemma}.
	Now, a fuzzy set $(A, \mu_A)$ satisfies $\vdash_{\iota} x\equiv y$ if and only if $\iota(x)=\iota(y)$, thus, if we take $(H, \id{H})$ and the assignment $\iota$ of the previous example, then $(H,\id{H})\vdash_{\iota} x\equiv y$ but it does not satisfy $ x\equiv z$ with respect to $\iota$.
\end{remark}

\subparagraph{Completeness} 
Now we prove that the calculus we have provided in \cref{sec:the} is complete.
Let us start with the following observation.
\begin{remark}\label{rel}
	For any $\Lambda\in \thr{L}$ the relation $\sim_{\Lambda}$ given by all $t$ and $s$ such that $ \vdash_{\Lambda} t \equiv s$,
	is an equivalence relation on $\terms{L}$.
\end{remark}
Using this equivalence, we can define the model of terms, as done next.
\begin{definition} 
	Let $\mathcal{L}=(\Sigma, X)$ be a language and $\Lambda\in\thr{L}$,  we define $\trms{\Lambda}$ to be the quotient of $\terms{L}$ by $\sim_{\Lambda}$, moreover, by rule \textsc{Fun}, the function
	\begin{align*}
		\hat{\mu}:\terms{L}\rightarrow H\qquad
		t \mapsto \sup\qty{l\in H \mid \hspace{2pt} \vdash_\Lambda \ex{l}{t}}
	\end{align*} 
	induces a function $\mu_\Lambda:\trms{\Lambda}\rightarrow H$.
	For any $f\in O$ and $c\in C$ putting $c^{\mathcal{T}_{\Lambda}}\ed[c]$ and
	\begin{align*}
		f^{\mathcal{T}_{\Lambda}}:\trms{\Lambda}^{\ar(f)} \rightarrow \trms{\Lambda}\qquad
		\qty([t_1],\dots, [t_{\ar(f)}])\mapsto \qty[f\qty(t_1,\dots, t_{\ar(f)})]	
	\end{align*} gives us a $\Sigma$-algebra $\mathcal{T}_\Lambda=\qty(\qty(\trms{\Lambda}, \mu_\Lambda), \Sigma^{\mathcal{T}_{\Lambda}})$, called \emph{ the $\Sigma$-algebra  of terms in $\Lambda$}.  
	The \emph{canonical assignment} is the function	$\iota_{can}:X\rightarrow \trms{\Lambda}$ sending $x$ to its class $[x]$.
\end{definition}
\begin{remark}
	Rule \textsc{Cong} assures us that $f^{\mathcal{T}_{\Lambda}}$ is well defined while  \textsc{Exp} implies that it is an arrow of $\fuz{H}$.
\end{remark}

\begin{lemma}\label{trms}Let $\mathcal{L}=(\Sigma, X)$ be a language and $\Lambda \in \thr{L}$, 	then:
	\begin{enumerate}	
		\item for any $\phi\in \stat{L}$ the following are equivalent:
		\begin{enumerate}
			\item  $\mathcal{T}_{\Lambda}\vDash \phi$,
			\item   $\mathcal{T}_{\Lambda}\vDash_{\iota_{can}}\phi$,
			\item $ \vdash_{\Lambda} \phi$;
		\end{enumerate}
		\item for any assignment $\iota:X \rightarrow \trms{\Lambda}$ and formula $\phi$, $\mathscr{T}_\Lambda \vDash_\iota \phi$ if and only if $ \vdash_\Lambda \phi[\sigma \circ \iota]$ for one (and thus any) section $\sigma$ of the quotient $\terms{L}\rightarrow \trms{\Lambda}$;
		\item $\mathcal{T}_\Lambda=\qty(\qty(\trms{\Lambda}, \mu_\Lambda), \Sigma^{\mathcal{T}_\Lambda})$ is a model of $\Lambda$.
	\end{enumerate}
\end{lemma}
\begin{proof}
		Let us start with a technical result.
		\begin{proposition}\label{ass}
			Let $\mathcal{L}=(\Sigma, X)$ be a language, $\Lambda \in \thr{L}$, and  $\sigma:\trms{\Lambda}\rightarrow \terms{L}$ a section of the quotient $\terms{L}\rightarrow \trms{\Lambda}$. 
			The equation $t^{\mathcal{T}_{\Lambda}, \iota}=\qty[t\qty[\sigma \circ \iota ]]$ holds for any assignment $\iota:X\rightarrow \trms{\Lambda}$ and $t\in \terms{L}$. In particular $t^{\mathcal{T}_{\Lambda}, \iota_{can}}=[t]$.
		\end{proposition}
		\begin{proof}
				For constants it is obvious, let us proceed by induction.
				\begin{itemize}
					\item For any $x\in X$, $\sigma\qty(\iota(x))$ is a representative of $\iota(x)$, so:
					\begin{align*}
						x^{\mathcal{T}_{\Lambda}, \iota}=\iota(x)=\qty[\sigma(\iota(x))]=\qty[\sigma\qty(\iota\qty(x))/x]=\qty[x\qty[\sigma\circ \iota]]
					\end{align*}
					\item For any $f\in O$ and $t_1,\dots,t_{\ar(f)}\in \terms{L}$:
					\begin{align*}
						&\qty(f\qty(t_1,\dots,t_{\ar(f)}))^{\mathcal{T}_\Lambda, \iota}=f^{\mathcal{T}_\Lambda}\qty(t^{\mathcal{T}_\Lambda,\iota}_1,\dots,t^{\mathcal{T}_\Lambda,\iota}_{\ar(f)})=f^{\mathcal{T}_\Lambda}\qty(t^{\mathcal{T}_\Lambda,\iota}_1,\dots,t^{\mathcal{T}_\Lambda,\iota}_{\ar(f)})\\&=f^{\mathcal{T}_\Lambda}\qty(\big{[}t_1\qty[\sigma\circ \iota]\big{]},\dots,\qty[t_{\ar(f)}\qty[\sigma\circ \iota]])=
						\qty[f\qty(t_1[\sigma\circ \iota],\dots,t_{\ar(f)}[\sigma\circ \iota])]\\&=
						\qty[\qty(f\qty(t_1,\dots,t_{\ar(f)}))[\sigma \circ \iota]]
						\qedhere 
					\end{align*}
				\end{itemize}
		\end{proof}
	Now we can proceed with the proof of \cref{trms}.
	\begin{enumerate}
		\item 	Let us show the three implications.
		\begin{enumerate}
			\item[](a)$\Rightarrow$(b). By definition.
			\item[](b)$\Rightarrow$(c). 
			We split the cases.
			\begin{itemize}
				\item $\phi$ is $t\equiv s$. Then $\mathcal{T}_\Lambda\vDash_{\iota_{can}} \phi$ means 
				\begin{align*}
					[t]=t^{\mathscr{T}_\Lambda, \iota_{can}}=s^{\mathscr{T}_\Lambda, \iota_{can}}=[s]
				\end{align*}
				thus $t\sim_{\Lambda}s$ i.e. $ \vdash_\Lambda t\equiv s$.
				\item $\phi$ is $\ex{l}{t}$. Let $S$ be $
					\qty{l'\in H \mid \Lambda \vdash \ex{l'}{t}}$,
				by hypothesis $\mathcal{T}_\Lambda\vDash_{\iota_{can}} \phi$, so 
				\begin{align*}
					l\leq \mu_\Lambda\qty(t^{\mathscr{T}_\Lambda, \iota_{can}})= \mu_\Lambda\qty([t])=\sup(S)
				\end{align*}
				hence $l=l\wedge \sup(S)$ and, since
				$H$ is a frame, it follows that $l= \sup \qty(\qty{l\wedge l'\mid l'\in S })$, by rule \textsc{Mon} we know that that $\vdash_\Lambda \ex{l\wedge l'}{t}$ for all $l'\in S$ and so rule \textsc{Sup} gives us $ \vdash_\Lambda \ex{l}{t}$.
			\end{itemize}	
			\item[](c)$\Rightarrow$(a). Let $\iota:X \rightarrow \trms{\Lambda}$ be an assignment and $\sigma$ a section as in the hypothesis, then we can apply \textsc{Sub} to get $\vdash_\Lambda \phi[\sigma\circ \iota]$, by \cref{ass} we get the thesis.
		\end{enumerate}
		\item By \cref{ass} we have
		\begin{align*}
			t^{\mathscr{T}_\Lambda, \iota}=[t[\sigma\circ \iota]]=\qty(t[\sigma\circ \iota])^{\mathscr{T}_\Lambda, \iota_{can}}
		\end{align*}
		we can conclude using the previous point.
		\item
		Let $\Gamma\vdash \psi$ be a sequent in $\Lambda$ with $\Gamma=\qty{\phi_i}_{i=1}^n$ and $\iota:X\rightarrow \trms{\Lambda}$ an assignment such that $\mathcal{T}_{\Lambda}\vDash_\iota \Gamma$. By point $1$ above this means that $\vdash_\Lambda \Gamma[\sigma \circ \iota]$ and applying \textsc{Sub} and \textsc{Cut} we can conclude that $\vdash_\Lambda \psi[\sigma \circ\iota]$.
		By the previous point this is equivalent to $\mathscr{T}_\Lambda\vDash_{\iota} \psi$.\qedhere
		\end{enumerate}
 
\end{proof}
\begin{corollary}[Completeness for formulae]
	 For any theory $\Lambda \in \thr{L}$,  $\mathcal{A}\vDash \phi$ for all $\mathcal{A}\in \mode{\Lambda}$ if and only if $\vdash_{\Lambda}\phi$. 
\end{corollary}

\section{From theories to monads}\label{sec:fre}
 Given a language $\mathcal{L}= \qty(\Sigma, X)$ and a fuzzy theory $\Lambda\in \thr{L}$ we have a forgetful functor: $\mathscr{U}_{\Lambda}:\mode{\Lambda}\rightarrow \fuz{L}$. In this section we first show that it has a left adjoint (\cref{sec:ffafs}) and that for a specific class of theories, models correspond to Eilenberg-Moore algebras for the monad induced by this adjunction (\cref{sec:emalg}).
	 
 \subsection{The free fuzzy algebra on a fuzzy set}\label{sec:ffafs}
To give the definition of free models (\cref{def:freemodel}) we need some preliminary constructions.
 
 \begin{definition}
 	Let $\mathcal{A}$ be a $\Sigma$-algebra and $f:(B,\mu_B)\rightarrow \mathscr{U}_\Sigma(\mathcal{A})$ an arrow in $\fuz{H}$, a $\Sigma$-algebra \emph{generated by $f$ in $\mathcal{A}$} is a morphism $\epsilon:\gen{B}{f}{A}\rightarrow \mathcal{A}$ such that:
 	\begin{itemize}
 		\item $\mathscr{U}_\Sigma(\epsilon)$ is strong mono; 
 		\item there exists $\bar{f}:(B,\mu_B)\rightarrow \gen{B}{f}{A}$ such that $\mathscr{U}_\Sigma(\epsilon) \circ \bar{f}=f$;
 		\item if $g:\mathcal{C}\rightarrow \mathcal{A}$ is a morphism such that $\mathscr{U}_\Sigma(g)$ is strong monomorphism  and
 		$\mathscr{U}_\Sigma(g)\circ h=f$ for some $h$ then there exists a unique $k:\gen{B}{f}{A}\rightarrow \mathcal{C}$ such that $g\circ k=\epsilon$.
 	\end{itemize}
 \end{definition}
	 
 \begin{proposition}\label{con}
 	For any signature $\Sigma$, $\Sigma$-algebra $\mathcal{A}$ and $f:(B, \mu_B)\rightarrow \mathscr{U}_\Sigma(\mathcal{A})$, $\gen{B}{f}{A}$ exists and it is unique up to isomorphism.
 \end{proposition}
 \begin{proof}
 	This is straightforward: if $\mathcal{A}=((A, \mu_A), \Sigma^{\mathcal{A}})$ it is enough to consider as carrier set $\mathscr{U}_\Sigma(\gen{B}{f}{A})$ the smallest subset of $A$ such that:
 	\begin{itemize}
 		\item $f(B)\subset  \mathscr{U}_\Sigma(\gen{B}{f}{A})$;
 		\item $\qty{c^{\mathcal{A}}}_{c\in C} \subset  \mathscr{U}_\Sigma(\gen{B}{f}{A})$;
 		\item if $g\in O$ and $b_1,\dots, b_{\ar(g)}\in \mathscr{U}(\gen{B}{f}{A})$ then $g^{\mathcal{A}}(b_1,\dots,b_{\ar(g)})\in\mathscr{U}_\Sigma(\gen{B}{f}{A})$. 
 	\end{itemize}
 	with the restriction of $\mu_A$ as membership degree function. As $\Sigma$-algebra structure we can take the restriction of $\Sigma^{\mathcal{A}}$. By \cref{mono} the inclusion $i$ is a strong monomorphism and by construction an arrow of $\alg{\Sigma}$. 
 	Let $g:\mathcal{C}\rightarrow \mathcal{A}$ and $h$ be as in the definition, if we consider the inclusion $i:(f(B), {\mu_A}_{|f(B)})\rightarrow (A,\mu_A)$, by the left lifting property we obtain
	\begin{center}
		\begin{tikzpicture}
	 		\node(A)at(0,0){$(B, \mu_B)$};
	 		\node(B)at(0,-1.5){$(f(B), {\mu_A}_{|f(B)})$};
	 		\node(C)at(3,-1.5){$(A,\mu_A)$};
	 		\node(D)at(3,0){$\mathscr{U}_\Sigma(\mathcal{C})$};
	 		\draw[->](A)--(B)node[pos=0.5, left]{$\bar{f}$};
	 		\draw[->](B)--(C)node[pos=0.5, below]{$i$};
	 		\draw[->](A)--(D)node[pos=0.5, above]{$h$};
	 		\draw[->](B)--(D)node[pos=0.5, above, xshift=0cm, yshift=0cm]{$\bar{k}$};
	 		\draw[<-](C)--(D)node[pos=0.5, right]{$\mathscr{U}_\Sigma(g)$};
			\end{tikzpicture}
	\end{center}
	and we can define $k$ by induction as:
	\begin{itemize}
		\item $k(b)=\bar{k}(b)$ if $b\in f(B)$;
		\item $k(c^{\mathcal{A}})=c^{\mathcal{C}}$ if $c\in C$;
		\item $k(t^{\mathcal{A}}(b_1,\dots,b_{\ar(g)}))=t^{\mathcal{C}}(k(b_1),\dots,k(b_{\ar(t)}))$ if $t\in O$ and $b_1,\dots,b_{\ar(t)}\in \gen{B}{f}{A}$.
	\end{itemize}
	Notice that injectivity of $g$ implies that $k$ is actually well-defined, in fact, proceeding by induction:
	\begin{itemize}
		\item if $x\in B$ is such that $f(x)=c^\mathcal{A}$ for some $c\in C$ then
			\begin{align*}
	 			g(h(x))=\epsilon(\bar{f}(x))=c^{\mathcal{A}}=g(c^\mathcal{C})
	 		\end{align*}
	 		so
	 		\begin{align*}
	 			k(f(x))=h(x)=c^{\mathcal{C}}=k(c^{\mathcal{A}})
	 		\end{align*}
 		 
	 	\item if $x\in B$ is such that $f(x)=t^{\mathcal{A}}(b_1,\dots, b_{\ar(t)})$ for some $t\in O$ and $b_1,\dots,b_{\ar(t)}\in \gen{B}{f}{A}$ then
	 		\begin{align*}
	 			g(h(x))&=\epsilon(\bar{f}(x))=t^{\mathcal{A}}(b_1, \dots, b_{\ar(t)})\\
	 			&=t^{\mathcal{A}}(g(k(b_1)), \dots, g(k(b_{\ar(t)})))= g(t^{\mathcal{C}}(b_1,\dots, b_{\ar(t)}))
	 		\end{align*}
	 		so $h(x)=t^{\mathcal{C}}(b_1,\dots, b_{\ar(t)})$ and we are done;
	 		
	 	\item  if $c\in C$ is such that $c^{\mathcal{A}}=t^{\mathcal{A}}(b_1,\dots, b_{\ar(t)})$ for some $t\in O$ and $b_1,\dots,b_{\ar(t)}\in \gen{B}{f}{A}$ then
	 		\begin{align*}
	 			g(c^{\mathcal{C}})&=c^{\mathcal{A}}=t^{\mathcal{A}}(b_1, \dots, b_{\ar(t)})=t^{\mathcal{A}}(g(k(b_1)), \dots, g(k(b_{\ar(t)})))=
	 			g(t^{\mathcal{C}}(b_1,\dots, b_{\ar(t)}))
	 		\end{align*}
	 		so $c^{\mathcal{C}}=t^{\mathcal{C}}(b_1,\dots, b_{\ar(t)})$ and we are done again; 
	 		
	 	\item if $s^{\mathcal{A}}(b'_1,\dots, b'_{\ar(s)})=t^{\mathcal{A}}(b_1,\dots, b_{\ar(t)})$ for some $t, s\in O$ and $b'_1,\dots,b'_{\ar(s)}, b_1,\dots,b_{\ar(t)}, \in \gen{B}{f}{A}$ then
	 		\begin{align*}
	 			g(s^{\mathcal{C}}(k(b'_1),\dots, k(b'_{\ar(s)}))&=s^{\mathcal{A}}(b'_1,\dots, b'_{\ar(s)})=t^{\mathcal{A}}(b_1, \dots, b_{\ar(t)})\\
	 			&=t^{\mathcal{A}}(g(k(b_1)), \dots, g(k(b_{\ar(t)})))= g(t^{\mathcal{C}}(b_1,\dots, b_{\ar(t)}))
	 		\end{align*}
	 		so $s^{\mathcal{C}}(k(b'_1),\dots, k(b'_{\ar(s)}))=t^{\mathcal{C}}(b_1,\dots, b_{\ar(t)})$ and even in this case $k$ is well defined.
	 	\end{itemize}
	 	Moreover, for any $t\in \gen{B}{f}{A}$, if $\mathscr{U}(\mathcal{C})=(D, \mu_D)$ then
	 	\begin{align*}
		 	\mu_{D}(k(t))=\mu_A(g(k(t)))=\mu_A(i(t))
	 	\end{align*}
	 	Uniqueness follows at once by induction.
 \end{proof}
	 
 \begin{remark}\label{sub2}
 	\cref{ind} implies that, given a model $\mathcal{A}=\qty(\qty(A, \mu_A), \Sigma^{\mathcal{A}})$  of a theory $\Lambda\in \thr{L}$, and a morphism $f:(B, \mu_B)\rightarrow (A, \mu_A)$, the $\Sigma$-algebra $\gen{B}{f}{A}$ is in $\mode{\Lambda}$.
 \end{remark}

 \begin{proposition}\label{uni} 
 	Let $\mathcal{A}$ be a $\Sigma$-algebra and $f:(B, \mu_B)\rightarrow \mathscr{U}_\Sigma(\mathcal{A})$, then, for any other $\Sigma$-algebra $\mathcal{C}$ and $h:(B, \mu_B)\rightarrow \mathscr{U}_\Sigma(\mathcal{C})$ there exists at most one $k:\gen{B}{f}{A}\rightarrow \mathcal{C}$ such that $k\circ \bar{f}=h$.
 \end{proposition}
 \begin{proof}
 	Let $k$ and $k'$ as in the thesis, we can proceed by induction:
	\begin{itemize}
		\item $k(f(x))=k'(f(x))$ for any $x\in B$ by hypothesis;
		\item $k\qty(c^{\gen{B}{f}{A}})=k'\qty(\gen{B}{f}{A})$ since $k$ and $k'$ are morphism of $\alg{\Sigma}$;
		\item $k\qty(g^{\gen{B}{f}{A}}\qty(b_1,\dots,b_{\ar(g)}))=g^{\mathcal{C}}\qty(k\qty(b_1),\dots,k(b_{\ar(g)}))$, by inductive hypothesis this is equal to  $g^{\mathcal{C}}(k'(b_1),\dots,k'(b_{\ar(g)}))$  and we conclude using again the fact that $k'$ is an arrow of $\alg{\Sigma}$.\qedhere
	 \end{itemize} 
\end{proof}

The next definition explains how to extend a theory in a given language with the data of a fuzzy set. 
\begin{definition}
	Let $(M, \mu_M)$ be a fuzzy set, $\mathcal{L}=(\Sigma, X)$ a language with $\Sigma = (O, \ar, C)$.
	We define $\Sigma[M, \mu_M]$ to be $\qty( O, \ar, C\sqcup M)$ and $\mathcal{L}_{(M,\mu_M)}$ to be $\qty(\Sigma[M, \mu_M], X)$.
	We have an obvious morphism $\mathbf{I}:\mathcal{L}\rightarrow \mathcal{L}_{(M,\mu_M)}$ given by the identities and the inclusion $i_C:C\rightarrow C\sqcup M$.
	
	For any $\Lambda \in \thr{L}$ we define $\Lambda[M, \mu_M]\in \mathcal{L}_{(M,\mu_M)}$ as $\mathbf{I}_*\qty(\Lambda)\cup \overline{\qty(M, \mu_M)}$ where $\overline{\qty(M, \mu_M)}=\qty{\vdash \ex{l}{m} \mid m \in M, l\in L \text{ and } \mu_M(m)\geq l}$.
\end{definition}
\begin{remark}
	It is immediate to see that $\mathbf{I^*}\qty(\Lambda[M,\mu_M])=\Lambda$.
\end{remark}

In the next proposition we show that, for any theory $\Lambda$, a fuzzy set can be mapped into the term model of the theory $\Lambda$ extended with it.
Hence, the natural candidate to be the free model is the algebra generated by such map.
 \begin{proposition}\label{init}
 	For any fuzzy set $(M,\mu_M)$ and any theory $\Lambda \in \thr{L}$:
 	\begin{enumerate}
 		\item the function $\bar{\eta}_{(M,\mu_M)}$ sending $m$ to the class $[m]$ of the corresponding constant defines an arrow of fuzzy sets $\bar{\eta}_{(M,\mu_M)}:(M,\mu_M)\rightarrow \trms{\Lambda\qty[M,\mu_M]}$;
 		\item any element in $\langle M,\mu_M \rangle_{\mathcal{T}_{\Lambda\qty[M,\mu_M]}, \bar{\eta}_{(M,\mu_M)}}$ has a representative without variables;
 		\item $\langle M,\mu_M \rangle_{\mathcal{T}_{\Lambda\qty[M,\mu_M]}, \bar{\eta}_{\qty(M,\mu_M)}}$ is the initial object of $\mode{\Lambda[M, \mu_M]}$;
 	\end{enumerate}
 \end{proposition}
 \begin{proof}
 	\begin{enumerate}
 		\item The only thing to do is to show that $\mu_M(m)\leq \mu_{{\Lambda[M,\mu_M]}}\qty([m])$	but we are easily done since $\vdash \ex{\mu_M}{m}\in \Lambda\qty[M, \mu_M]$.
	 	\item This follows by the explicit construction given in \cref{con}. 
	 	\item $\langle M,\mu_M \rangle_{\mathcal{T}_{\Lambda\qty[M,\mu_M]}, \bar{\eta}_{\qty(M,\mu_M)}}$ is in $\mode{\Lambda[M, \mu_M]}$ by \cref{sub2}.
	 	For every model $\mathcal{A}$ of $\Lambda[M, \mu_M]$, mapping $m\in M$ to $m^{\mathcal{A}}$ provides a morphism of fuzzy sets $(M,\mu_M)\rightarrow \mathscr{U}_{\Sigma[M,\mu_M]}(\mathcal{A})$. By \cref{uni} there exists at most one morphism of $\Sigma[M,\mu_M]$-algebras  from $\langle M,\mu_M \rangle_{\mathcal{T}_{\Lambda\qty[M,\mu_M]}, \bar{\eta}_{\qty(M,\mu_M)}}$ to $\mathcal{A}$.
	 	On the other hand, by point 2  for every $[t]\in \langle M,\mu_M \rangle_{\mathcal{T}_{\Lambda\qty[M,\mu_M]}, \bar{\eta}_{(M,\mu_M)}}$ there exists a term $s$ with no variables such that $\Lambda[M, \mu_M]\vdash t\equiv s$, therefore $t^{\mathcal{A}, \iota}=t^{\mathcal{A}, \iota'}$ for any two assignments $\iota, \iota':X\rightarrow A$. Hence we get a morphism of $\Sigma[M,\mu_M]$-algebras sending the class $[t]$ of a term $t$ to its interpretation $t^{\mathcal{A}, \iota}$ for any chosen assignment $\iota$.\qedhere
	\end{enumerate}
 \end{proof}
	 
 \begin{definition}\label{def:freemodel}
 	For any language $\mathcal{L}$, $\Lambda\in \thr{L}$ and $(M,\mu_M)\in \fuz{H}$ we define the \emph{free model $\mathcal{F}_\Lambda\qty(M,\mu_M)$ of $\Lambda$ generated by $\qty(M,\mu_M)$} to be $\res{I}{\Lambda\qty[M,\mu_M]}\qty(\langle M,\mu_M \rangle_{\mathcal{T}_{\Lambda\qty[M,\mu_M]},\bar{\eta}_{\qty(M,\mu_M)}})$.
 	We set $\term{\Lambda}\qty(M,\mu_M)$ to be $\mathscr{U}_{\Lambda}\qty(\mathcal{F}_\Lambda(M,\mu_M))$.
 \end{definition}
Now it is enough to check that the free model just defined actually provides the left adjoint.%
 \begin{theorem}\label{free}
 	For any language $\mathcal{L}$ and $\Lambda\in \thr{L}$ the functor $\mathscr{U}_\Lambda:\mode{\Lambda}\rightarrow \fuz{L}$ has a left adjoint $\mathscr{F}_{\Lambda}$.
 \end{theorem}
 \begin{proof}
	By construction $\bar{\eta}_{(M,\mu_M)}$ factors through $	 	\eta_{(M,\mu_M)}:(M,\mu_M)\rightarrow \term{\Lambda}(M,\mu_M)$ which sends $m$ to $[m]$.
	Let now $g:(M,\mu_M)\rightarrow \mathscr{U}_\Lambda(\mathcal{B})$ be another arrow in $\fuz{H}$, we can turn $\mathcal{B}$ into a $\Sigma[M,\mu_M]$-algebra $\mathcal{B}^g$ setting $m^{\mathcal{B}^g}$ to be $g(m)$ for any $m\in M$.
	
    Let us show that $\mathcal{B}^g$ is a model of $\Lambda[M, \mu_M]$. 
	 Surely it is a model of $\Lambda$ since $\mathcal{B}$ is, let $\vdash \ex{l}{m}$ be a sequent in $\overline{(M, \mu_M)}$, then for any assignment $\iota:V\rightarrow B$:
	 	\begin{align*}
	 		\mathcal{B}^g \vDash_\iota  \ex{l}{m} &\iff l\leq \mu_{B}(m^{\mathcal{B}^g, \iota})l\leq \mu_{\Lambda}\qty(
	 		t^{\mathscr{F}_{\Lambda}(M, \mu_M), \eta_{(M,\mu_M)}\circ \iota})\iff l\leq \mu_B(g(m)) 
	 	\end{align*}
	 	but $g$ is an arrow of $\fuz{H}$ so $\mu_B(g(m))\geq \mu_M(m)$ and we are done.

	 Now, since $\mathcal{B}^{g}$ is a model of $\Lambda[M, \mu_M]$, we can take $\bar{g}$ to be the image through $\res{\mathbf{I}}{\Lambda[M, \mu_M]}$ of the unique arrow $\langle M,\mu_M \rangle_{\mathcal{T}_{\Lambda[(M,\mu_M)]}, \bar{\eta}_{(M,\mu_M)}} \rightarrow \mathcal{B}^{g}$, by construction
	 \begin{align*}
	 	\bar{g}(\eta_{(M,\mu_M)}(m))=\bar{g}([m])=m^{\mathcal{B}^g}=g(m)
	 \end{align*}
	 so $\mathscr{U}_\Lambda(\bar{g})\circ \eta_{(M,\mu_M)}=g$. Uniqueness follows from \cref{init}.
\end{proof} 

\begin{definition}
 	Given a theory $\Lambda\in \thr{L}$, its \emph{associated monad} $\term{\Lambda}:\fuz{H}\rightarrow \fuz{H}$ is the composite $\mathscr{U}_\Lambda\circ \mathscr{F}_\Lambda$. 
\end{definition}
 
\begin{remark}
 	If $\Lambda$ is the empty theory (in any language), then, by \cref{left}, the composittion $\mathscr{F}_\emptyset\circ \nabla$ gives us a functor isomorphic to $\trma{\Sigma}$.
\end{remark}
\begin{notazione}
	We will denote by $\trm{\emptyset}$ with $\trm{\Sigma}$ and with  $\term{\Sigma}$ the monad $\term{\emptyset}=\mathscr{U}_\Sigma\circ \trm{\Sigma}$.	
\end{notazione}

In this setting we can provide a result similar to \cref{trms}.
\begin{lemma}\label{trms2}
 	For any language  $\mathcal{L}=(\Sigma, X)$ we define the \emph{natural assignment} $\iota_{nat}$ as the adjoint to the unit $\nabla(X)\rightarrow \term{\Lambda}(\nabla(X))$. Then $\trm{\Lambda}(\nabla(X))\vDash_{\iota_{nat}}\phi$ if and only if $\vdash_\Lambda \phi $.
\end{lemma}
\begin{proof}
	The implication from the right to the left follows immediately since $\trm{\Lambda}(\nabla(X))$ is a model for $\Lambda$. By adjointness he canonical assignment $\iota_{can}$ induces an arrow $\nabla(X)\rightarrow \mathcal{U}_{\Lambda[\nabla(X)]}\qty( \mathcal{T}_{\Lambda[\nabla(X)]})$, which, in turn, induces a morphism $e:\trm{\Lambda}(\nabla(X))\rightarrow \mathcal{T}_{\Lambda[\nabla(X)]}$ of $\Sigma$-algebras such that, as function between sets, $e\circ \iota_{nat}=\iota_{can}$. Recalling that $\mathbf{I}$ is the arrow $(\Sigma, X)\rightarrow (\Sigma[\nabla(X)], X )$ and using \cref{ind}, \cref{mrp} and \cref{trms}: 
	\begin{align*}
		\trm{\Lambda}(\nabla X)\vDash_{\iota_{nat}}\phi &\iff \res{I}{\Lambda[\nabla(X)]}\qty(\gn)\vDash_{\iota_{nat}}\phi \\
		&\iff\rs{I}\qty(\gn)\vDash_{\iota_{nat}}\phi \\
		&\iff \gn \vDash_{\iota_{nat}}\phi \\ 
		&\hspace{0.2cm}\Longrightarrow  \mathcal{T}_{\Lambda[\nabla(X)]} \vDash_{e\circ \iota_{nat}}\phi\\
		&\iff \mathcal{T}_{\Lambda[\nabla(X)]} \vDash_{\iota_{can}}\phi \\&\iff \vdash_{\Lambda[\nabla(X)]} \phi
	\end{align*}
	Now, by definition $\overline{\nabla(X)}$ is equal to $\qty{\vdash \ex{\bot}{x}\mid x\in X}$, therefore $\qty(\Lambda[\nabla(X)])^{\vdash}=\Lambda^{\vdash}$ and we get the thesis. 
\end{proof}
	 
\subsection{Eilenberg-Moore algebras and models}\label{sec:emalg}
In this section we will compare the category $\mode{\Lambda}$ of models of some $\Lambda \in \thr{L}$ and $\eim{\Lambda}$ of Eilenberg-Moore algebras for the corresponding monad $\term{\Lambda}$. First of all we recall the following classic lemma (\cite[Prop.~4.2.1]{borceux1994handbook} and \cite[Theorem VI.3.1]{mac2013categories}).
\begin{lemma}
 	Let $\mathscr{L}:\catname{C}\rightarrow \catname{D}$ be a functor with right adjoint $\mathscr{R}$ and let $\mathsf{T}=\mathscr{R}\circ \mathscr{L}$ be its associated monad, then there exists a \emph{comparison functor} $\mathscr{K}:\catname{D}\rightarrow \catname{Alg}(\mathsf{T})$ such that 
 	$\mathscr{U}_{\mathsf{T}}\circ \mathscr{K}=\mathscr{R}$, where $\mathscr{U}_{\mathsf{T}}:\catname{Alg}(\mathsf{T})\rightarrow \catname{C}$ is the forgetful functor.
 	$\mathscr{K}$ sends $D$ in $(\mathscr{R}(D), \mathscr{R}(\epsilon_D))$, where $\epsilon$ is the counit of the adjunction.
\end{lemma}
As a consequence, for any theory $\Lambda$ we have a functor from $\mode{\Lambda}$ to $\eim{\Lambda}$. We want to construct an inverse of such functor.
\begin{definition}
	Let $\Lambda$ be in $\thr{L}$ and $\xi:\term{\Lambda}(M, \mu_M)\rightarrow (M, \mu_M)$ an object of $\eim{\Lambda}$, we define its \emph{associated algebra} $\mathscr{H}(\xi)=\qty(\qty(M, \mu_M), \Sigma^{\mathscr{H}(\xi)})$ putting
	\begin{equation*}
		c^{\mathscr{H}(\xi)}\ed\xi\qty(c^{\mathscr{F}_\Lambda(X, \mu_X)}) \qquad f^{\mathscr{H}\qty(\xi)}\ed\xi \circ f^{\mathscr{F}_\Lambda(X, \mu_X)}\circ \eta_{\qty(M,\mu_M)}^{\ar(f)}
	\end{equation*}
	 for every $c\in C$ and $f\in O$.
\end{definition}

\begin{lemma}\label{hom}
 	For any Eilenberg-Moore algebra $\xi:\term{\Lambda}(M, \mu_M)\rightarrow (M, \mu_M)$, $\xi$ itself is an arrow $\mathscr{F}_\Lambda(X,\mu_X)\rightarrow \mathscr{H}(\xi)$ of $\alg{\Sigma}$.
\end{lemma}
\begin{proof}
	Let us start with the following observation.		
	\begin{claim}\label{expl}
		Let $\Lambda$ be a theory in the language $\mathcal{L}$ and $\hat{\mu}$ the multiplication of $\term{\Lambda}$, then, for any  $g:(A,\mu_A)\rightarrow (B, \mu_B)$ and operation $f$ the following diagrams commutes:
	 	\begin{center}
	 		\begin{tikzpicture}
	 			\node(A) at(0,-1.5) {$\term{\Lambda}(B, \mu_B	)^{\ar(f)}$};
	 			\node(B) at(4,-1.5) {$\term{\Lambda}(B, \mu_B)$};
	 			\node(D) at(4,0) {$\term{\Lambda}(M, \mu_M)$};
	 			\node(C) at(0,0) {$\term{\Lambda}(M, \mu_M)^{\ar(f)}$};
	 			\draw[->](A)--(B) node[pos=0.5, below]{$f^{\mathscr{F}_{\Lambda}(B, \mu_B)}$};
	 			\draw[<-](B)--(D) node[pos=0.5, right]{$\term{\Lambda}(g)$};
	 			\draw[<-](A)--(C) node[pos=0.5, left]{$\term{\Lambda}(g)^{\ar(f)}$};
	 			\draw[->](C)--(D) node[pos=0.5, above]{$f^{\mathscr{F}_\Lambda(M, \mu_M)}$};
	 				
 				\node(A) at(-0.5,1) {$\term{\Lambda}(M, \mu_M	)^{\ar(f)}$};
 				\node(B) at(4.5,1) {$\term{\Lambda}(M, \mu_M	)^{\ar(f)}$};
 				\node(D) at(4.5,2.5) {$\term{\Lambda}(\term{\Lambda}(M, \mu_M))$};
 				\node(C) at(-0.5,2.5) {$\term{\Lambda}(\term{\Lambda}(M, \mu_M))^{\ar(f)}$};
 				\draw[->](A)--(B) node[pos=0.5, above]{$f^{\mathscr{F}_\Lambda(M, \mu_M)}$};
 				\draw[<-](B)--(D) node[pos=0.5, left]{$\hat{\mu}_{(M,\mu_M)}$};
 				\draw[<-](A)--(C) node[pos=0.5, left]{$\hat{\mu}_{(M,\mu_M)}^{\ar(f)}$};
 				\draw[->](C)--(D) node[pos=0.5, above]{$f^{\mathscr{F}_\Lambda\qty(\term{\Lambda}(M, \mu_M))}$};
	 				
 				\node(A) at(8.95,2.5) {$(1, c_{\bot})$};
 				\node(B) at(7.2,1) {$\term{\Lambda}\qty(\term{\Lambda}(M, \mu_M))$};
 				\node(C) at(10.7,1) {$\term{\Lambda}(M, \mu_M)$};
 				\draw[->](A)--(B) node[pos=0.5, left]{$c^{\mathscr{F}_{\Lambda}\qty(\term{\Lambda}\qty(M, \mu_M))}$};
 				\draw[->](B)--(C) node[pos=0.5, above]{$\hat{\mu}_{(M,\mu_M)}$};
 				\draw[->](A)--(C) node[pos=0.5, right, xshift=0.2cm]{$c^{\mathscr{F}_{\Lambda}(M, \mu_M)}$};
 					 				
	 			\node(A) at(8.95,0) {$(1, c_{\bot})$};
	 			\node(B) at(7.2,-1.5) {$\term{\Lambda}(M, \mu_M)$};
	 			\node(C) at(10.7,-1.5) {$\term{\Lambda}(M, \mu_M)$};
	 			\draw[->](A)--(B) node[pos=0.5, left]{$c^{\mathscr{F}_{\Lambda}(M, \mu_M)}$};
	 			\draw[->](B)--(C) node[pos=0.5, below]{$\term{\Lambda}(g)$};
	 			\draw[->](A)--(C) node[pos=0.5, right, xshift=0.2cm]{$c^{\mathscr{F}_{\Lambda}(M, \mu_M)}$};
	 		\end{tikzpicture}
	 	\end{center}	
	\end{claim}
	\begin{proof}
		$\term{\Lambda}(g)=\mathscr{U}_\Lambda(\mathscr{F}_\Lambda(g))$ and  $\hat{\mu}_{(A,\mu_A)} = \mathscr{U}_\Lambda(\epsilon_{\mathscr{F}_\Lambda(A,\mu_A)})$  where $\epsilon:\mathscr{F}_\Lambda\circ \mathscr{U}_\Lambda\rightarrow \id{\mode{\Lambda}}$ is the counit (\cite{borceux1994handbook} proposition $4.2.1$ or \cite{mac2013categories}, chapter VI). Now the thesis follows since both $\mathscr{F}_\Lambda(g)$ and $\epsilon_{\mathscr{F}_\Lambda(A,\mu_A)} $ are arrows of $\alg{\Sigma}$.
	\end{proof}
	We can now come back to the proof of our lemma: $\xi\circ c^{\mathscr{F}_\Lambda(M,\mu_M)}=c^{\mathscr{H}(\xi)}$ while the other condition is equivalent to commutativity of the outer rectangle in the diagram:
	\begin{center}
	 		\begin{tikzpicture}[scale=0.85]
	 		\node(A) at (0,0){$\term{\Lambda}\qty(\term{\Lambda}(M,\mu_M))^{\ar(f)}$};
	 		\node(B) at (0,2){$\term{\Lambda}(M,\mu_M)^{\ar(f)}$};
	 		\node(C) at (0,-2.5){$\term{\Lambda}(M,\mu_M)^{\ar(f)}$};
	 		\node(D) at (-4,0){$\term{\Lambda}(M,\mu_M)^{\ar(f)}$};
	 		\node(E) at (-4,2){$(M,\mu_M)^{\ar(f)}$};
	 		\node(F) at (5.5,-2.5){$\term{\Lambda}(M,\mu_M)$};
	 		\node(G) at (9.5,-2.5){$(M,\mu_M)$};
	 		\node(H) at (5.5,0){$\term{\Lambda}\qty(\term{\Lambda}(M,\mu_M))$};
	 		\node(I) at (9.5,0){$\term{\Lambda}(M,\mu_M)$};
	 		\draw[->](E)--(B)node[pos=0.5, above]{$\eta_{(M, \mu_M)}^{\ar(f)}$};
	 		\draw[->](D)--(A)node[pos=0.5, below]{$\eta_{\term{\Lambda}(M, \mu_M)}^{\ar(f)}$};
	 		\draw[->](D)--(E)node[pos=0.5, left]{$\xi^{\ar(f)}$};
	 		\draw[->](A)--(B)node[pos=0.5, right]{$\term{\Lambda}(\xi)^{\ar(f)}$};
	 		\draw[->](A)--(H)node[pos=0.5, below]{$f^{\mathscr{F}_\Lambda\qty(\mathscr{F}_\Lambda\qty(M,\mu_M))}$};
	 		\draw[->](A)--(C)node[pos=0.5, right]{$\hat{\mu}_{(M,\mu_M)}^{\ar(f)}$};
	 		\draw[->](C)--(F)node[pos=0.5, below]{$f^{\mathscr{F}_\Lambda\qty(M,\mu_M)}$};
	 		\draw[->](H)--(F)node[pos=0.5, left]{$\hat{\mu}_{(M,\mu_M)}$};
	 		\draw[->](F)--(G)node[pos=0.5, below]{$\xi$};
	 		\draw[->](I)--(G)node[pos=0.5, right]{$\xi$};
	 		\draw[->](H)--(I)node[pos=0.5, above]{$\term{\Lambda}(\xi)$};
	 		\draw[->](B)..controls(6.5,2) and (9.5,2)..(I)node[pos=0.2, above, yshift=0cm]{$f^{\mathscr{F}_\Lambda\qty(M,\mu_M)}$};
	 		\draw[->](D)..controls (-4,-2.5)..(C)node[pos=0.2, left, yshift=0cm]{$\id{\term{\Lambda}(M, \mu_M)}$};
	 		\node(a)at(-2,1){$1$};
	 		\draw (a) circle [radius=0.3cm];
	 		\node(b)at(-2,-1.25){$3$};
	 		\draw (b) circle [radius=0.3cm];
	 		\node(c)at(2.75,-1.25){$4$};
	 		\draw (c) circle [radius=0.3cm];
	 		\node(d)at(7.5,-1.25){$5$};
	 		\draw (d) circle [radius=0.3cm];
	 		\node(d)at(4.75,1){$2$};
	 		\draw (d) circle [radius=0.3cm];
	 		\end{tikzpicture}
	\end{center}
	but $\circled{1}$ commutes by naturality of $\eta$, $\circled{2}$ and $\circled{4}$ by \cref{expl}, $\circled{3}$ since $\term{\Lambda}$ is a monad and $\circled{5}$ from the fact that $\xi$ is an Eilenberg-Moore algebra.
\end{proof}

In general $\mathscr{H}(\xi)$ is not a model of $\Lambda$, but we can restrict to a class of theories such this holds. As in \cite{bacci2020quantitative, mardare2017axiomatizability}, we consider theories whose sequents' premises contain only variables.
\begin{definition}
	A theory $\Lambda\in \thr{L}$ is \emph{basic}\footnote{In \cite{bacci2018algebraic} such theories are called \emph{simple}.} if, for any sequent $\Gamma \vdash \phi$ in it, all the formulae in $\Gamma$ contain only variables.
\end{definition}

\begin{example}
	Fuzzy groups, fuzzy normal groups, fuzzy semigroups and left, right, bilateral ideals (\cref{ex:fuzzyg1,ex:semig1}) are all examples of basic theories.
\end{example}

\begin{lemma}$ \mathscr{H}(\xi)$ is a model of $\Lambda$ for any basic theory
$\Lambda\in \thr{L}$ and Eilenberg-Moore algebra $\xi:\term{\Lambda}(M, \mu_M)\rightarrow (M, \mu_M)$.
\end{lemma}
\begin{proof}
		We can start by observing that if $\Gamma \vdash \phi$ is in $\Lambda$ and $\iota:X\rightarrow M$ is an assignment such that $	\mathscr{H}(\xi)\vDash_{\iota} \Gamma$ then $\mathscr{F}_{\Lambda}(M, \mu_M)\vDash_{\eta_{(M,\mu_M)}\circ \iota} \Gamma$.
	Indeed, for every $\psi$  in $\Gamma$, we have two cases:
	\begin{itemize}
		\item $\psi$ is $x\equiv y$. Since $\iota(x)=\iota(y)$ we can easily conclude.
		\item $\psi$ is $\ex{l}{x}$. The thesis follows at once since the membership degree of $\eta_{(M,\mu_M)}(\iota(x))$ in $\term{\Lambda}(M, \mu_M)$ is greater than $\mu_{M}(\iota(x))$.
	\end{itemize}

Therefore, we know  that $\mathscr{F}_{\Lambda}(M, \mu_M)\vDash_{\eta_{(M,\mu_M)}\circ \iota} \phi$. Let us split again the two cases.
	\begin{itemize}
 		\item $\phi$ is $t\equiv s$. In this case, $
 		t^{\mathscr{F}_{\Lambda}(M, \mu_M), \eta_{(M,\mu_M)}\circ \iota}=s^{\mathscr{F}_{\Lambda}(M, \mu_M), \eta_{(M,\mu_M)}\circ \iota}$, point $2$ of \cref{ind} and the fact that $\xi$ is an Eilenberg-Moore algebra thus imply:
 		\begin{align*}
	 		t^{\mathscr{H}(\xi),  \iota}&=t^{\mathscr{H}(\xi), \xi \circ  \eta_{(M,\mu_M)}\circ \iota}=\xi\qty(t^{\mathscr{F}_{\Lambda}(M, \mu_M), \eta_{(M,\mu_M)}\circ \iota})\\&=\xi\qty(s^{\mathscr{F}_{\Lambda}(M, \mu_M), \eta_{(M,\mu_M)}\circ \iota})=s^{\mathscr{H}(\xi), \xi \circ  \eta_{(M,\mu_M)}\circ \iota}=s^{\mathscr{H}(\xi), \iota}
 		\end{align*}
 		\item $\phi$ is $\ex{l}{t}$. This means that $l\leq \mu_{\Lambda}\qty(t^{\mathscr{F}_{\Lambda}(M, \mu_M), \eta_{(M,\mu_M)}\circ \iota})$, hence, using again \cref{hom} and \cref{ind}:
 		\begin{align*}
	 		l&\leq \mu_{\Lambda}\qty(
	 		t^{\mathscr{F}_{\Lambda}(M, \mu_M), \eta_{(M,\mu_M)}\circ \iota})\leq \mu_{M}\qty(\xi\qty(
	 		t^{\mathscr{F}_{\Lambda}(M, \mu_M), \eta_{(M,\mu_M)}\circ \iota}))\\&=
	 		\mu_M\qty(t^{\mathscr{H}(\xi), \xi \circ  \eta_{(M,\mu_M)}\circ \iota})=\mu_M\qty(t^{\mathscr{H}(\xi),  \iota})
 		\end{align*}
 		and we can conclude.\qedhere 
 	\end{itemize} 	 
 \end{proof}
 
 \begin{theorem} For any basic theory $\Lambda\in \thr{L}$, the functor $\mathscr{K}:\mode{\Lambda}\rightarrow \eim{\Lambda}$ has an inverse $\mathscr{H}:\eim{\Lambda}\rightarrow \mode{\Lambda}$ sending $\xi:\term{\Lambda}(M, \mu_M)\rightarrow (M, \mu_M)$ to $\mathscr{H}(\xi)$. 
	 \end{theorem}
	 \begin{proof}
	 	Let $\xi:\term{\Lambda}(M, \mu_M)\rightarrow (M, \mu_M)$, $\xi':\term{\Lambda}(N, \mu_N)\rightarrow (N, \mu_N)$ be two Eilenberg-Moore algebras and $g:(M, \mu_M)\rightarrow (N, \mu_N)$ an arrow between them, we claim that $g$ itself is a morphism of $\mode{\Lambda}$. Let $f\in O$ and $c\in C$, we have diagrams:
	 	\begin{center}
	 		\begin{tikzpicture}
	 		\node(A) at(3.5,2) {$(M, \mu_M	)$};
	 		\node(B) at(7.5,2) {$(N, \mu_N)$};
	 		\node(D) at(7.5,3.5) {$\term{\Lambda}(N, \mu_N)$};
	 		\node(C) at(3.5,3.5) {$\term{\Lambda}(M, \mu_M)$};
	 		\draw[->](A)--(B) node[pos=0.5, below]{$g$};
	 		\draw[<-](B)--(D) node[pos=0.5, right]{$\xi'$};
	 		\draw[<-](A)--(C) node[pos=0.5, left]{$\xi$};
	 		\draw[->](C)--(D) node[pos=0.5, below]{$\term{\Lambda}(g)$};
	 		\node(d)at(5.5,2.65){$2$};
	 		\draw (d) circle [radius=0.2cm];
	 		\node(e)at(5.5,4.2){$1$};
	 		\draw (e) circle [radius=0.2cm];
	 		\node(E) at(5.5,5) {$(1, c_{\bot})$};
	 		\draw[->](E)--(C) node[pos=0.5, left]{$c^{\mathscr{F}_{\Lambda}(M, \mu_M)}$};
	 		\draw[->](E)--(D) node[pos=0.5, right, xshift=0.2cm]{$c^{\mathscr{F}_{\Lambda}(N, \mu_N)}$};
	 		
	 		\node(A) at(0,1){$(M,\mu_M)^{\ar(f)}$};
	 		\node(B) at(0,-0.5){$(N,\mu_N)^{\ar(f)}$};
	 		\node(C) at(3.5,1){$\term{\Lambda}(M,\mu_M)^{\ar(f)}$};
	 		\node(D) at(3.5,-0.5){$\term{\Lambda}(N,\mu_N)^{\ar(f)}$};
	 		\node(E) at(7.5,1){$\term{\Lambda}(M,\mu_M)$};
	 		\node(F) at(7.5,-0.5){$\term{\Lambda}(N,\mu_N)$};
	 		\node(G) at(11,1){$(M,\mu_M)$};
	 		\node(H) at(11,-0.5){$(N,\mu_N)$};
	 		\node(a)at(1.75,0.25){$3$};
	 		\draw (a) circle [radius=0.2cm];
	 		\node(b)at(5.5,0.25){$4$};
	 		\draw (b) circle [radius=0.2cm];
	 		\node(c)at(9.25,0.25){$5$};
	 		\draw (c) circle [radius=0.2cm];
	 		\draw[->](A)--(B)node[pos=0.5, left]{$g^{\ar(f)}$};
	 		\draw[->](C)--(D)node[pos=0.5, right]{$\term{\Lambda}(g)^{\ar(f)}$};
	 		\draw[->](E)--(F)node[pos=0.5, left]{$\term{\Lambda}(g)$};
	 		\draw[->](G)--(H)node[pos=0.5, right]{$g$};
	 		\draw[->](A)--(C)node[pos=0.5, above]{$\eta_{(M, \mu_M)}^{\ar(f)}$};
	 		\draw[->](B)--(D)node[pos=0.5, below]{$\eta_{(N, \mu_N)}^{\ar(f)}$};
	 		
	 		\draw[->](C)--(E)node[pos=0.5, above]{$f^{\mathscr{F}_\Lambda(M,\mu_M)}$};
	 		\draw[->](D)--(F)node[pos=0.5, below]{$f^{\mathscr{F}_\Lambda(N,\mu_N)}$};
	 		\draw[->](E)--(G)node[pos=0.5, above]{$\xi$};
	 		\draw[->](F)--(H)node[pos=0.5, below]{$\xi'$};
	 		\end{tikzpicture}
	 	\end{center}
	 	Commutativity of $\circled{1}$ and $\circled{4}$ follows from \cref{expl}, that of $\circled{3}$ from naturality of $\eta$, $\circled{2}$ and $\circled{5}$ from the fact that $g$ is an arrow of $\eim{\Lambda}$. So $\mathscr{H}$ is a functor. Since both $\mathscr{H}$ and $\mathscr{K}$ are the identity on arrows, it is enough to show that they are the inverse of each other on objects.
	 	Let $\xi:\term{\Lambda}(M,\mu_M)\rightarrow (M,\mu_M)$ be in $\eim{\Lambda}$,
	 	then \cref{hom} implies that $\xi$ is a morphism of $\mode{\Lambda}$ such that $\mathscr{U}_\Lambda(\xi)\circ \eta_{(M,\mu_M)}=\id{\qty(M,\mu_M)}$, but there is only one such morphism, namely the component of the counit in $\mathscr{H}(\xi)$, so
	 	$\mathscr{U}_\Lambda\qty(\epsilon_{\mathscr{H}(\xi)})=\xi$ and $\mathscr{K}\circ \mathscr{H}=\id{\eim{\Lambda}}$. On the other hand, let $\mathcal{A}=\qty(\qty(A,\mu_A), \Sigma^{\mathcal{A}})\in \mode{\Lambda}$, and consider, for any $f\in O$ and $c\in C$, the diagrams
	 	\begin{center}
	 		\begin{tikzpicture}
	 		\node(A)at (0,0) {$\term{\Lambda}(A,\mu_A)^{\ar(f)}$};
	 		\node(B)at (0,-1.5) {$(A,\mu_A)^{\ar(f)}$};
	 		\node(C)at (4,0) {$\term{\Lambda}(A,\mu_A)$};
	 		\node(D)at (4,-1.5) {$(A,\mu_A)$};
	 		\node(E)at(-4,0){$(A,\mu_A)^{\ar(f)}$};
	 		\draw[->](E)--(A)node[pos=0.5, above]{$\eta_{(A, \mu_A)}^{\ar(f)}$};
	 		\draw[->](E)..controls(-4,-1.5)..(B)node[pos=0.5, left]{$\id{{(A, \mu_A)}^{\ar(f)}}$};
	 		\draw[->](A)--(B)node[pos=0.5, right]{$\mathscr{U}_\Lambda(\epsilon_{\mathcal{A}})^{\ar(f)}$};
	 		\draw[->](C)--(D)node[pos=0.5, right]{$\mathscr{U}_\Lambda(\epsilon_{\mathcal{A}})$};
	 		\draw[->](A)--(C)node[pos=0.5, above]{$f^{\mathcal{A}}$};
	 		\draw[->](B)--(D)node[pos=0.5, below]{$f^{\mathcal{A}}$};
	 		\node(a)at(-2,-0.75){$2$};
	 		\draw (a) circle [radius=0.2cm];
	 		\node(b)at(2.5,-0.75){$3$};
	 		\draw (b) circle [radius=0.2cm];
	 		
	 		\node(A)at (0,2.5) {$(1,c_{\bot})$};
	 		\node(C)at (-2,1) {$\term{\Lambda}(A,\mu_A)$};
	 		\node(D)at (2,1) {$(A,\mu_A)$};
	 		\draw[->](C)--(D)node[pos=0.5, below]{$\mathscr{U}_\Lambda(\epsilon_{\mathcal{A}})$};
	 		\draw[->](A)--(C)node[pos=0.5, left]{$c^{{\mathscr{F}_\Lambda(A, \mu_A)}}$};
	 		\draw[->](A)--(D)node[pos=0.5, right, xshift=0.1cm]{$c^{\mathcal{A}}$};
	 		\node(b)at(0,1.75){$1$};
	 		\draw (b) circle [radius=0.2cm];
	 		\end{tikzpicture}
	 	\end{center}
	 	Commutativity of $\circled{1}$ and $\circled{3}$ follows since each component of $\epsilon$ is an arrow of $\mode{\Lambda}$, that of $\circled{2}$ since $\epsilon$ is the counit. So we can deduce now that $
	 	f^{\mathscr{H}(\mathscr{U}_\Lambda(\epsilon_{\mathcal{A}}))}=f^{\mathcal{A}}$ and $ c^{\mathscr{H}(\mathscr{U}_\Lambda(\epsilon_{\mathcal{A}}))}=c^{\mathcal{A}}$ 
	 	from which we can deduce that $\mathscr{H}\circ \mathscr{K}=\id{\mode{\Lambda}}$. 
 \end{proof}
 \begin{corollary}
 	For any basic theory $\Lambda\in \thr{L}$, $\eim{\Lambda}$ and $\mode{\Lambda}$ are isomorphic, and thus equivalent, categories.
 \end{corollary}
	 
 \section{Equational axiomatizations}\label{sec:birk}
In this section we prove two results for our calculus analogous to the classic HSP theorem \cite{birkhoff1935structure}, using the results by Milius and Urbat \cite{milius2019equational}.

\subparagraph{The abstract framework} 
Let us start recalling the tools introduced in \cite{milius2019equational}, adapted to our situation. 
In the following we will fix a tuple\footnote{In their work Milius and Urbat additionaly require a full subcategory of $\catname{C}$ and a fixed class of cardinals, but we will not need this level of generality.} $\qty(\catname{C}, \qty(\mathscr{E}, \mathscr{M}), \mathscr{X} )$ where $\catname{C}$ is a category, $\qty(\mathscr{E}, \mathscr{M})$ is a proper factorization system on $\catname{C}$ and  $\mathscr{X}$ is a class of objects of $\catname{C}$.

 \begin{definition}
 	An object $X$ of $\catname{C}$ is \emph{projective with respect to an arrow $f:A\rightarrow B$} if for any $h:X\rightarrow B$ there exists a $k:X\rightarrow A$ such that $f\circ k=h$.
 	We define $	\mathscr{E}_{\mathscr{X}}$ as the class of $e\in\mathscr{E}$ such that for every $X\in \mathscr{X}$, $X$ is projective with respect to $e$. A \emph{$\mathscr{E}_{\mathscr{X}}$-quotient} is just an arrow in $\mathscr{E}_{\mathscr{X}}$.
\end{definition}
In the rest of the section, we assume that  $\qty(\catname{C},\qty(\mathscr{E}, \mathscr{M}), \mathscr{X} )$ satisfies the following requirements:
 	\begin{itemize}
 		\item $\catname{C}$ has all (small) products;
 		\item for any $X\in \mathscr{X}$, the class $
 		\arr{X}{C}$ of all $e\in \mathscr{E}$ with domain $X$	is essentially small, i.e. there is a set $\mathcal{J}\subset\arr{X}{C}$ such that for any $e:X\rightarrow C\in \arr{X}{C}$ there exists $e':X\rightarrow D \in \mathcal{J}$ and an isomorphism $\phi$ such that $\phi \circ e=e'$;

 		\item for every object $C$ of $\catname{C}$ there exists $e:X\rightarrow C$ in $\mathscr{E}_{\mathscr{X}}$ with $X\in \mathscr{X}$.
	\end{itemize}

 \begin{definition}\label{equat} 
 	An \emph{$\mathscr{X}$-equation} is an arrow $e\in \arr{X}{C}$ with $X\in \mathscr{X}$.
 	We say that an object $A$ of $\catname{C}$ \emph{satisfies}  $e:X\rightarrow C$, and we write $A\vDash_{\mathscr{X}}e$, if for every $h:X\rightarrow A$ there exists  $q:C\rightarrow A$ such that $q\circ e=h$.  Given a class $\mathbb{E}$ of $\mathscr{X}$-equations, we define $\mathcal{V}(\mathbb{E})$ as the full subcategory of $\catname{C}$ given by objects that satisfy $e$ for every $e\in \mathbb{E}$. A full subcategory $\catname{V}$ is $\mathscr{X}$-\emph{equationally presentable} if there exists $\mathbb{E}$ such that $\catname{V}=\mathcal{V}(\mathbb{E})$.
\end{definition}
\begin{remark}\label{rmk}
 	The definition of equation in \cite[Def.~3.3]{milius2019equational} is given in terms of suitable subclasses of $\arr{X}{C}$. However in our setting Milius and Urbat's definition reduces to ours (cfr.~\cite[Remark 3.4]{milius2019equational}).
 \end{remark}

 \begin{theorem}[{\cite[Th.~3.15, 3.16]{milius2019equational}}] \label{milius} 
 	A full subcategory $\catname{V}$  of $\catname{C}$ is $\mathscr{X}$-equationally presentable if and only if it is closed under $\mathscr{E}_\mathscr{X}$-quotients, $\mathscr{M}$-subobjects and (small) products.
 \end{theorem} 
 
 \subparagraph{Application to fuzzy algebras}
 In order to apply the results above to $\alg{\Sigma}$, we need to define the required inputs, i.e., to specify a factorization system and a class of $\Sigma$-algebras.
 \begin{lemma}\label{facto}
 	For any signature $\Sigma$ the classes $\mathscr{E}_\Sigma\ed\qty{e \text{ arrow of } \alg{\Sigma} \mid \mathscr{U}_\Sigma(e) \text{ is epi} }$ and 
 	$\mathscr{M}_\Sigma\ed\qty{m \text{ arrow of } \alg{\Sigma} \mid \mathscr{U}_\Sigma(m) \text{ is strong mono} }$ form a proper factorization system on $\alg{\Sigma}$.
 \end{lemma}
 \begin{proof}
 	Let $\mathcal{A}=\qty((A, \mu_A), \Sigma^{\mathcal{A}})$ and $\mathcal{B}=\qty((B, \mu_B), \Sigma^{\mathcal{B}})$ two $\Sigma$-algebras with a morphism $g:\mathcal{A}\rightarrow \mathcal{B}$ between them. $\mathcal{U}_\Sigma(g)$ factors as $m\circ e$ where $e:(A, \mu_A)\rightarrow(g(A), {\mu_B}_{|g(A)})$ and $m:(g(A), {\mu_B}_{|g(A)})\rightarrow(B, \mu_B)$ is the usual epi-monomorphism factorization of $f$ on $\fuz{H}$ (cfr.~\cref{prod}). 
 	Notice that  $c^{\mathcal{B}}=g(c^\mathcal{A})\in g(A)$ for all $c\in C$ and $f^{\mathcal{B}}\qty(g(a_1),\dots, g(a_{\ar(g)}))=g\qty(f^{\mathcal{A}}(a_1,\dots,a_{\ar(g)}))\in g(A)$ for every $f\in O$  and $g(a_1),\dots, g(a_{\ar(g)})\in g(A)$ so $\Sigma^{\mathcal{B}}$ restricts to a $\Sigma$-algebra structure on $(g(A), {\mu_B}_{|g(A)})$ and it is clear that with this choice $e$ and $m$ becomes morphisms of $\alg{\Sigma}$. We have now to show the left lifting property. 
 	Given $e\in \mathscr{E}_\Sigma$, $m\in \mathscr{M}_\Sigma$ and $g$ and $h$ such that $m\circ g=h\circ e$ we can apply $\mathscr{U}_\Sigma$ and get a square
	 	\begin{center}
	 		\begin{tikzpicture}
	 		\node(A) at(0,0) {$\mathscr{U}_\Sigma(\mathcal{A})$};
	 		\node(B) at(2.5,0) {$\mathscr{U}_\Sigma(\mathcal{B})$};
	 		\node(D) at(2.5,-1.5) {$\mathscr{U}_\Sigma\qty(\mathcal{B'})$};
	 		\node(C) at(0,-1.5) {$\mathscr{U}_\Sigma\qty(\mathcal{A'})$};
	 		\draw[->](A)--(B) node[pos=0.5, above]{$\mathscr{U}_\Sigma(g)$};
	 		\draw[->](B)--(D) node[pos=0.5, right]{$\mathscr{U}_\Sigma(m)$};
	 		\draw[->](A)--(C) node[pos=0.5, left]{$\mathscr{U}_\Sigma(e)$};
	 		\draw[->](C)--(D) node[pos=0.5, below]{$\mathscr{U}_\Sigma(h)$};
	 		\end{tikzpicture}
	 	\end{center}
	 	which, by \cref{prod}, has a diagonal filling $k:\mathscr{U}_\Sigma\qty(\mathcal{A}')\rightarrow \mathscr{U}_\Sigma(\mathcal{B})$. Let us show that $k$ is a morphism of $\alg{\Sigma}$. For any $c\in C$ we have $c^{\mathcal{A'}}=e(c^\mathcal{A})$ and for every $f\in O$ and $x_1,\dots, x_{\ar(f)}\in \mathscr{U}_\Sigma\qty(\mathcal{A}')$ there exist $a_1,\dots, a_{\ar(f)}\in \mathscr{U}_\Sigma\qty(\mathcal{A})$ such that $e(a_i)=x_i$ for $1\leq i\leq \ar(f)$, so $k(c^\mathcal{A'})=k\qty(e\qty(c^\mathcal{A}))=g\qty(c^\mathcal{A})=c^\mathcal{B}$ and
\begin{align*}
	 	k\qty(f^{\mathcal{A'}}\qty(x_1,\dots, x_{\ar(f)})) & = 
	    k\qty(f^{\mathcal{A'}}\qty(e\qty(a_1),\dots,e\qty(a_{\ar(f)})))=
	    k\qty(e\qty(f^{\mathcal{A}}\qty(a_1,\dots,a_{\ar(f)}))) \\
	    & = g\qty(f^{\mathcal{A}}\qty(a_1,\dots,a_{\ar(f)})) = 
	    f^{\mathcal{B}}\qty(g\qty(a_1),\dots,g\qty(a_{\ar(f)})) \\
	    &= f^{\mathcal{B}}\qty(k\qty(e\qty(a_1)),\dots,k\qty(e\qty(a_{\ar(f)}))) =f^{\mathcal{B}}\qty(k\qty(x_1),\dots,k\qty(x_{\ar(f)}))
\end{align*}

	 	Finally, properness follows at once since $\mathscr{U}_\Sigma$ is faithful and so reflects epis and monos. 
	 \end{proof}

 \begin{definition}
 We define the following two classes of $\Sigma$-algebras: 
\begin{align*}
\mathscr{X}_0\ed\qty{\mathscr{F}^{\catname{Set}}_{\Sigma}\qty(X)\mid X \in \catname{Set}}\qquad
\mathscr{X}_{\mathsf{E}}\ed\qty{\mathscr{F}_{\Sigma}(X, \mu_X)\mid (X, \mu_X)\in \fuz{H}}
\end{align*} 
We will use $\mathscr{E}_{\Sigma,\mathscr{X}_0}$ (resp.,  $\mathscr{E}_{\Sigma,\mathscr{X}_\mathsf{E}}$) for the class of $e\in\mathscr{E}$ such that every $X\in \mathscr{X}_0$ (resp. $X\in \mathscr{X}_\mathsf{E}$) is projective with respect to $e$.
\end{definition}
 \begin{remark}
 	$\mathscr{X}_0=\qty{\trm{\Sigma}\qty(X, \mu_X)\mid \supp{X}=\emptyset}$. 
 \end{remark}
 We  have now all the ingredients needed to use the results recalled above.
 \begin{lemma}\label{aaa}
 	With the above definitions:
 	\begin{enumerate}
 		\item $\mathscr{E}_{\Sigma,\mathscr{X}_0}=\mathscr{E}_\Sigma$;
 		\item $	\mathscr{E}_{\Sigma, \mathscr{X}_{\mathsf{E}}}=\qty{e\in \mathscr{E}_{\Sigma} \mid \mathscr{U}_\Sigma(e) \text{ is split}}$;
 		\item  $(\alg{\Sigma}, (\mathscr{E}_\Sigma, \mathscr{M}_\Sigma), \mathscr{X}_0)$ and $(\alg{\Sigma}, (\mathscr{E}_\Sigma, \mathscr{M}_\Sigma), \mathscr{X}_{\mathsf{E}} )$ satisfy the conditions of our settings.
	\end{enumerate}
\end{lemma}
\begin{proof}
	 	Let us start adapting the usual notion of congruence to our set environment.
	 	\begin{definition}
	 		Given a $\Sigma$-algebra $\mathcal{A}=\qty(\qty(A, \mu_A), \Sigma^{\mathcal{A}})$, a \emph{fuzzy congruence} on $\mathcal{A}$ is a pair $(\sim, \mu)$ where
	 		\begin{itemize}
	 			\item $\sim$ is a congruence: i.e. an equivalence relation such that, for any $f\in O$, if $a_i\sim b_i$ fo $1\leq i \leq \ar(f)$ then $f^{\mathcal{A}}\qty(a_1,\dots a_{\ar(f)})\sim f^{\mathcal{A}}\qty(b_1,\dots, b_{\ar(f)})$;
	 			\item $\mu:A\rightarrow H$ is a function such $\mu(a)=\mu(b)$ whenever  $a\sim b$;
	 			\item for any $f\in O$,
	 			\begin{equation*}
	 				\bigwedge_{i=1}^{\ar(f)} \mu(a_i)  \leq \mu\qty(f^{\mathcal{A}}\qty(a_1,\dots a_{\ar(f)}))
	 			\end{equation*}
	 			\item  $\mu_A(a)\leq \mu(a)$ for every $a\in A$.
	 		\end{itemize}
	 	\end{definition}
	 	\begin{proposition}\label{cong}
	 		Let $\mathcal{A}=\qty(\qty(A, \mu_A), \Sigma^{\mathcal{A}})$ be a $\Sigma$-algebra, then:
	 		\begin{enumerate}[label=(\alph*)]
	 			\item if $\qty{\qty(\sim_{i}, \mu_i)}_{i\in I}$ is a family of fuzzy congruence then $\qty(\bigcap_{i\in I}\sim_{i}, \mu)$ with 
	 			$\mu:A\rightarrow H$ the pointwise infimum of $\qty{\mu_i}_{i\in I} $
	 			is a fuzzy congruence;
	 			\item for every fuzzy congruence $(\sim, \mu)$ there exists an epimorphism: $e_{(\sim, \mu)}:\mathcal{A}\rightarrow \mathcal{B}$ such that $\mu_B(b)=\mu(a)$ for any $a\in e^{-1}(b)$ and $e_{(\sim, \mu)}(a)=e_{(\sim, \mu)}(b)$ if and only if $a\sim b$;
	 			\item for every epimorphism $e:\mathcal{A}\rightarrow \mathcal{B}$ there exists a fuzzy congruence $(\sim_e, \mu_e)$ on $\mathcal{A}$ such that $e\leq e_{(\sim, \mu)}$ and $e_{(\sim, \mu)}\leq e$ in $\mathcal{A}\arro \alg{\Sigma}$.
	 		\end{enumerate}
	 	\end{proposition}
	 	\begin{proof}
	 	\begin{enumerate}[label=(\alph*)]
	 		\item This is straightforward.
	 		\item Define $\mathcal{B}=\qty(\qty(B, \mu_B), \Sigma^{\mathcal{B}})$ setting $B:=A/\sim$, $\mu_B([a]):=\mu(a)$ and, for any $f\in O$
	 		\begin{equation*}
	 		f^{\mathcal{A}}\qty(\qty[a_1],\dots,[a_{\ar(f)}]):=\qty[f^{\mathcal{B}}\qty(\qty[a_1], \dots, [a_{\ar(f)}])]
	 		\end{equation*}
	 		Since $(\sim, \mu)$ is a fuzzy congruence all these objects are well defined, the fact that $f^{\mathcal{B}}$ is an arrow of fuzzy sets follows from the second condition on $\mu$, while the last condition entails that the projection on the quotient is an arrow of $\fuz{H}$.
	 		\item Put $a\sim_e b$ if and only if $e(a)=e(b)$ and $\mu_e(a):=\mu_B(e(a))$. Since $e$ is a morphism of $\alg{\Sigma}$ we get the first and the last condition in the definition of a fuzzy congruence, while the second one follows since
	 		\begin{align*}
	 		\mu_e\qty(f^{\mathcal{A}}\qty(a_1,\dots, a_{\ar(f)}))&=\mu_B\qty(e\qty(f^{\mathcal{A}}\qty(a_1,\dots, a_{\ar(f)})))=\mu_B\qty(f^{\mathcal{B}}\qty(e\qty(a_1),\dots,  e\qty(a_{\ar(f)})))\\&\geq \bigwedge_{i=1}^{\ar(f)}\mu_B\qty(e\qty(a_i))=\bigwedge_{i=1}^{\ar(f)}\mu_e\qty(a_i)		 \end{align*}
	 		Now, it is immediate to see that the function sending the equivalence class $[a]$ of $a\in A$ to $e(b)$
	 		induces an isomorphism of $\alg{\Sigma}$ witnessing the thesis.\qedhere 
	 	\end{enumerate}
 	\end{proof}
	 	So equipped we can turn back to the proof of \cref{aaa}.
	 	\begin{enumerate}
	 		\item Let  $e:\mathcal{A}\rightarrow \mathcal{B}$ be an arrow in $\mathscr{E}_\Sigma$ and let  $h:\trma{\Sigma}(X)\rightarrow \mathcal{B}$ be any morphism of $\alg{\Sigma}$. By point $2$ of  \cref{mono} $e$ is surjective so for any $x\in X$ we can take a $a_x\in e^{-1}\qty(h\qty(\eta_X\qty(x)))$, where $\eta$ is the unit of the adjunction $\trma{\Sigma}\dashv\mathscr{V}_\Sigma $ of \cref{left}, and define $
	 		\bar{k}:X\rightarrow A$ mapping $x$ to $a_x$,  
	 		where $\mathcal{A}=\qty(\qty(A, \mu_A), \Sigma^{\mathcal{A}})$. By adjointness, from $\bar{k}$ we get $k:\trma{\Sigma}(X)\rightarrow \mathcal{A}$ and
	 		\begin{align*}
	 		\qty(e\circ  k)\circ \eta_{X}=e\circ \qty(k\circ \eta_{X})=e \circ \bar{k}=h\circ \eta_{X} 
	 		\end{align*}
	 		so $e\circ k = h$.
	 		\begin{center}
	 			\begin{tikzpicture}
	 			\node(A) at(1.5,0){$\mathscr{V}_\Sigma(\mathcal{A})$};
	 			\node(B) at(1.5,-1.5){$\mathscr{V}_\Sigma(\mathcal{B})$};
	 			\node(X) at(-2,-1.5){$\mathscr{V}_\Sigma\qty(\trma{\Sigma}\qty(X))$};
	 			\node(Y) at(-5.5,-1.5){$X$};
	 			\draw[->](A)--(B)node[pos=0.5, right]{$\mathscr{V}_\Sigma(e)$};
	 			\draw[<-](X)--(Y)node[pos=0.5, below]{$\eta_{\nabla(X)}$};
	 			\draw[->](X)--(B)node[pos=0.5, below]{$\mathscr{V}_\Sigma(h)$};
	 			\draw[->, dashed](X)..controls(-1.75,-0.75)and(-1.25, -0.25)..(A)node[pos=0.35, left, xshift=-0.1cm]{$\mathscr{V}_\Sigma(k)$};
	 			\draw[->, dashed](Y)..controls(-4.5,0)and(-1, 0)..(A)node[pos=0.5, above]{$\bar{k}$};
	 			\end{tikzpicture}
	 		\end{center}
	 		\item Let $e:\mathcal{A}\rightarrow \mathcal{B}$ be in $\mathscr{E}_\Sigma$ such that $\mathscr{U}_\Sigma(e)$ is split and let $s$ be a section in $\fuz{H}$, then, for any $h:\trm{\Sigma}(X,\mu_X)\rightarrow \mathcal{B}$ we can consider the arrow $s\circ h\circ \eta_{(X, \mu_X)}$, which, by adjointness provides a $k:\trm{\Sigma}(X,\mu_X)\rightarrow \mathcal{A}$, moreover:
	 		\begin{align*}
	 		\qty(e\circ k)\circ \eta_{(X, \mu_X)}&=e\circ \qty(k\circ \eta_{(X, \mu_X)})=e\circ \qty(s\circ h\circ \eta_{(X,\mu_X)})\\&=\qty(e\circ s)\circ \qty(h\circ \eta_{(X,\mu_X)})=h\circ \eta_{(X,\mu_X)} 
	 		\end{align*} 
	 		so $k$ is the wanted lifting. On the other hand, if $e$ is in $\mathscr{E}_{\Sigma, \mathscr{X}_1}$ we can take the diagram:
	 		\begin{center}
	 			\begin{tikzpicture}
	 			\node(A) at(1.5,0){$\mathscr{U}_\Sigma(\mathcal{A})$};
	 			\node(B) at(1.5,-1.5){$\mathscr{U}_\Sigma(\mathcal{B})$};
	 			\node(X) at(-2,-1.5){$\mathscr{U}_\Sigma\qty(\trm{\Sigma}(\mathscr{U}_\Sigma(\mathcal{B})))$};
	 			\node(Y) at(-5.5,-1.5){$\mathscr{U}_\Sigma(\mathcal{B})$};
	 			\draw[->](A)--(B)node[pos=0.5, right]{$\mathscr{U}_\Sigma(e)$};
	 			\draw[<-](X)--(Y)node[pos=0.5, above]{$\eta_{\mathscr{U}_\Sigma(\mathcal{B})}$};
	 			\draw[->](X)--(B)node[pos=0.5, above]{$\mathscr{U}_\Sigma(\epsilon_{\mathcal{B}})$};
	 			\draw[->](X)..controls(-1.75,-0.25)and(0, 0)..(A)node[pos=0.5, above, xshift=-0.05cm]{$\mathscr{U}_\Sigma(k)$};
	 			\draw[->](Y)..controls(-2.5,-2.5)and(-1.5, -2.5)..(B)node[pos=0.5, below,]{$\id{\mathscr{U}_\Sigma(\mathcal{B})}$};
	 			\end{tikzpicture}
	 		\end{center}
	 		where $\epsilon_{\mathcal{B}}$ is the component of the counit $\epsilon:\trm{\Sigma}\circ \mathscr{U}_\Sigma\rightarrow \id{\alg{\Sigma}}$ and $k$ its lifting. Taking $\mathscr{U}_\Sigma(k)\circ \eta_{\mathscr{U}_\Sigma(\mathcal{B})}$ we get the desired section of $\mathscr{U}_\Sigma(e)$.
	 		
	 		\item Let us proceed by steps.
	 		\begin{itemize}
	 			\item $\fuz{H}$ has all products by \cref{prod}.
	 			\item  $\arr{X}{C}$ is essentially small by point $3$ of \cref{cong}.
	 			\item For any fuzzy set $(X, \mu_X)$ we can consider the identity $
	 			\id{(X, \mu_X)}:(X, \mu_X)\rightarrow (X, \mu_X)$ and the counit $\epsilon_{(X, \mu_X)}:\nabla(X)\rightarrow (X, \mu_X)$ of the adjunction $\nabla \dashv\mathscr{U}$ of \cref{adj}. They induce arrows $
	 			e_0:\trma{\Sigma}(X)\rightarrow (X, \mu_X)$ and $ e_{\mathsf{E}}:\trm{\Sigma}(X, \mu_X)\rightarrow (X, \mu_X)$
	 			such that $\mathscr{U}_\Sigma \qty(e_0)\circ \eta_{\nabla(X)}= \epsilon_{(X, \mu_X)}$ and $\mathscr{U}_\Sigma \qty(e_H)\circ \eta_{(X, \mu_X)}=\id{(X, \mu_X)}$. So $\mathscr{U}_\Sigma \qty(e_H)$ is split and, since $\epsilon_{(X, \mu_X)}$ is surjective, point $2$ of \cref{mono} allows us to conclude that $\mathscr{U}_\Sigma \qty(e_0)$ is an epimorphism. \qedhere  
	 		\end{itemize}  
	 	\end{enumerate}
	 \end{proof}
 
 We want now to translate formulae of our sequent calculus into $\mathscr{X}_0$- and $\mathscr{X}_\mathsf{E}$-equations.  To this end, we have to restrict to two classes of theories, which we introduce next. 
	 \begin{definition}
	 	Let $\mathcal{L}=(\Sigma, X)$ be a language, a $\Lambda \in \thr{L}$ is said to be:
	 	\begin{itemize}
	 		\item  \emph{unconditional} (\cite[App.~B.5]{milius2019equational}) if any sequent in $\Lambda$ is of the form $\vdash \phi$ for some formula $\phi$;
	 		\item \emph{of type $\mathsf{E}$} if any sequent in $\Lambda$ is of the form $\qty{\ex{l_i}{x_i}}_{i\in I}\vdash \phi$ for some formula $\phi$, $\{x_i\}_{i\in I}\subset X$  and $\{l_i\}_{i\in I}\subset H$.
	 	\end{itemize}
	 \end{definition}
 
	 \begin{lemma}\label{te}  For any signature $\Sigma$ and $\mathscr{X}_{\mathsf{E}}$-equation $e:\trm{\Sigma}(X, \mu_X)\rightarrow \mathcal{B}$ there exists a type ${\mathsf{E}}$ theory $\Lambda_e$ such that, for every $\Sigma$-algebra $\mathcal{A}$, $	 	\mathcal{A}\vDash_{\mathscr{X}_1} e$ if and only if $\mathcal{A}\in \modd(\Lambda_e)$. Moreover $\abs{\Gamma}\leq \abs{\supp{X}}$ for any $\Gamma \vdash \phi \in \Lambda_e$.
	 \end{lemma}
	 \begin{proof}
	 	Let $\mathcal{L}_e$ be $(\Sigma, X)$. We define $\Gamma_X:=\{\ex{\mu_X(x)}{x}\mid x\in \supp{X}\}$
	 	and  $\Lambda_e\in \thr{L}$ as $\Lambda_e^{1}\cup \Lambda_e^{2}$ where
	 	\begin{align*}
	 	\Lambda_e^{1}\ed\qty{\Gamma_X\vdash \ex{l}{t}\mid l\leq \mu_{B}\qty(e\qty([t]))}\qquad 
	 	\Lambda_e^{2}\ed\qty{\Gamma_X\vdash [s]\equiv [t]\mid e\qty([t])=e\qty([s])}	
	 	\end{align*}
	 	and $(B, \mu_B)$ is $\mathscr{U}_\Sigma(\mathcal{B})$. Let us show the two implications.
	 	\begin{itemize}
	 		\item[$\Rightarrow$] Any $\iota:X\rightarrow A$ such that $\mathcal{A}\vDash_\iota \Gamma_X$ defines an arrow $\bar{\iota}(X,\mu_X)\rightarrow \mathscr{U}_\Sigma(\mathcal{A})$. By adjointness we have a homomorphism $h_\iota:\trm{\Sigma}(X,\mu_X)\rightarrow \mathcal{A}$ hence, by hypothesis, there exists $q_\iota:\mathcal{B}\rightarrow \mathcal{A}$ such that $q_\iota\circ e=h_\iota$. Now, notice that (see~\cref{free}, and \cref{init}(4)) $h_\iota([t])=t^\mathcal{A, \iota}$ Take a sequent $\Gamma_X \vdash \psi$ in $\Lambda_e$, we have two cases, depending on $\psi$.
	 		\begin{itemize}
	 			\item If $\psi=\ex{l}{ t}\in \Lambda^{me}_{e}$ we have
	 			\begin{align*}
	 			l\leq \mu_B\qty(e\qty([t]))\leq \mu_A\qty(q_\iota\qty(e\qty([t])))=\mu_A\qty(h_\iota\qty([t]))={t}^{\mathcal{A}, \iota}
	 			\end{align*} 
	 			therefore $\mathcal{A}\vDash_{\iota}\psi$. 
	 			\item If $\phi=[s]\equiv [t]\in \Lambda^{eq}_{e}$ then
	 			\begin{align*}
	 			{t}^{\mathcal{A}, \iota}=h_\iota\qty([t])=q_\iota\qty(e([t]))=q_\iota\qty(e([s]))=h_\iota\qty([t])={s}^{\mathcal{A}, \iota}
	 			\end{align*} 
	 			and even in this case  $\mathcal{A}\vDash_{\iota}\psi$.  
	 		\end{itemize}
	 		\item[$\Leftarrow$]	Take $h:\trm{\Sigma}\qty(X,\mu_X)\rightarrow \mathcal{A}$, $\mathscr{U}_\Sigma(h)\circ \eta_{\nabla(X)}$ is an arrow $(X,\mu_X)\rightarrow \mathscr{U}_\Sigma(\mathcal{A})$, so forgetting the fuzzy set structure too gives us an assignment $\iota_h:X\rightarrow A$ such that $\mathcal{A}\vDash_{\iota_h} \Gamma_X$. As before $
	 		h([t])=t^{\mathcal{A}, \iota_h}
	 	$
	 		for every $[t]\in \trm{\Sigma}\qty(X,\mu_X)$. Since $\mathcal{A}\in \mode{\Lambda_e}$ we have
	 		\begin{itemize}
	 			\item $t^{\mathcal{A}, \iota_h}=s^{\mathcal{A}, \iota_h}$ for all terms $t$ and $s$ such that $e([t])=e([s])$;
	 			\item $l\leq \mu_A(t^{\mathcal{A}, \iota_h})$ for all terms $t$ such that $l\leq \mu_B(e([t]))$.
	 		\end{itemize}
	 		So, the function $
	 		q:B\rightarrow A$ which sends $
	 		b\in B$ to  $h([t])$ for some $[t]\in e^{-1}(b)$, provides us with an arrow $\mathscr{U}_\Sigma(\mathcal{B})\rightarrow \mathscr{U}_\Sigma(\mathcal{A})$ such that
	 		$q\circ e=h$. Now:
	 		\begin{equation*}
	 		\begin{split}
	 		&q\qty(c^{\mathcal{B}})=q\qty(e\qty(c^{\trm{\Sigma}\qty(X,\mu_X)}))\\&=h\qty(c^{\trm{\Sigma}\qty(X,\mu_X)})=c^{\mathcal{A}}
	 		\\ \\ 
	 		\end{split}\quad
	 		\begin{split}
	 		&q\qty(f^{\mathcal{B}}\qty(b_1,\dots, b_{\ar(f)}))=q\qty(f^{\mathcal{B}}\qty(e\qty(c_1),\dots, \qty(e\qty(c_{\ar(f)}))))\\&=q\qty(e\qty(f^{\trm{\Sigma}\qty(X,\mu_X)}\qty(c_1,\dots,c_{\ar\qty(f)} )))\\&=h\qty(f^{\trm{\Sigma}\qty(X,\mu_X)}\qty(c_1,\dots,c_{\ar\qty(f)} ))=
	 		f^{\mathcal{A}}\qty(h\qty(c_1),\dots, h\qty(c_{\ar(f)}))\\&=
	 		f^{\mathcal{A}}\qty(q\qty(e\qty(c_1)),\dots, q\qty(e\qty(c_{\ar(f)})))=
	 		f^{\mathcal{A}}\qty(q\qty(b_1),\dots, q\qty(b_{\ar(f)}))
	 		\end{split} 
	 		\end{equation*}
	 		so we can conclude that $q$ is an arrow of $\alg{\Sigma}$ and we are done.	 		\qedhere 
	 	\end{itemize} 
	 \end{proof}
 
	 \begin{corollary}\label{tec}For any signature $\Sigma$ and $\mathscr{X}_0$-equation $e:\trma{\Sigma}\qty(X)\rightarrow \mathcal{B}$ there exists an unconditional theory  $\Lambda_e$ such that, for any $\Sigma$-algebra $\mathcal{A}$, $	 	\mathcal{A}\vDash_{\mathscr{X}_0} e$ if and only if $\mathcal{A}\in \modd(\Lambda_e)$.
	 \end{corollary}
	 
	 Finally, from the results above we can easily conclude HSP-like results for $\alg{\Sigma}$.	 
	 \begin{theorem}\label{hsp} Let $\catname{V}$ be a full subcategory  of $\alg{\Sigma}$, then
	 	\begin{enumerate}
	 		\item    $\catname{V}$ is closed under epimorphisms, (small) products and strong monomorphisms if and only if there exists a class of unconditional theories $\qty{\Lambda_e}_{e\in \mathbb{E}}$  such that $\mathcal{A}\in \catname{V}$ if and only if $\mathcal{A}\in \mode{\Lambda_e}$  for all $e\in \mathbb{E}$.
	 		
	 		\item    $\catname{V}$ is closed under split epimorphisms, (small) products and strong monomorphisms if and only if there exists a class of type $\mathsf{E}$ theories $\qty{\Lambda_e}_{e\in \mathbb{E}}$  such that $\mathcal{A}\in \catname{V}$ if and only if $\mathcal{A}\in \mode{\Lambda_e}$  for all $e\in \mathbb{E}$.
	 	\end{enumerate}
	 \end{theorem}
	 \begin{proof}
	 	This is straightforward in light of \cref{milius}, \cref{te} and \cref{tec}.
	 \end{proof} 
 
	 \begin{remark}
	 	We cannot arrange the collection $\{\Lambda_e\}_{e\in \mathbb{E}}$ into a unique theory since a proper class of variables is needed to write down all the necessary sequents. A possible way to deal with this issue is to fix two Grothendieck universes (\cite{williams1969grothendieck}) $\catname{U}_1\subset \catname{U}_2$ and modify the definition of language allowing for a possible class (i.e., an element of $\catname{U}_2$) of variables. All the proof of this paper can be repeated verbatim in this context carefully distinguishing between fuzzy \emph{sets} (i.e., those defined on an element of $\catname{U}_1$) and fuzzy \emph{classes} (i.e., those defined on an element of $\catname{U}_2$). Then the algebras of terms will be a fuzzy class in general but it is possible to show, using the explicit construction, that $\term{\Lambda}(X,\mu_X)$ is a fuzzy set if $X\in\catname{U}_1$ and so we can retain all the results of  \cref{sec:fre}.
	 \end{remark}
	 
	The issue mentioned in the previous remark can be avoided if the family $\{\Lambda_e\}_{e\in\mathbb{E}}$ satisfies a boundedness property
	about the premises of the sequents belonging to each $\Lambda_e$.	 
	 \begin{definition}
	 	Given a cardinal $\kappa$ we say that a $\mathscr{X}_\mathsf{E}$-equation $e:\trm{\Sigma}(X,\mu_X)\rightarrow \mathcal{B}$ is \emph{$\kappa$-supported} if $\abs{\supp{X}}<\kappa$.
	 \end{definition}
	 \begin{proposition}
	 	Let $\catname{V}=\mathcal{V}(\mathbb{E})$ be an $\mathcal{X}_\mathsf{E}$-equational defined subcategory of $\alg{\Sigma}$ and suppose every $e\in \mathbb{E}$ is $\kappa$-supported, then there exists a theory $\Lambda\in \thr{L}$, where $\mathcal{L}=(\Sigma, \kappa)$, such that $\catname{V}=\mode{\Lambda}$.
	 \end{proposition}
	 \begin{proof}For any $e:\trm{\Sigma}(X_e, \mu_{X_e})\rightarrow \mathcal{B}_e$ in $\mathbb{E}$ we can fix an injection $i_e:\supp{X_e}\rightarrow \kappa$ and an extension let $\bar{i}_e:X\rightarrow \kappa$ of it, fix also  morphisms $\mathbf{I}^e:\mathcal{L}_e\rightarrow\mathcal{L}$ given by $(\id{\Sigma}, \bar{i}_e)$. Let now $\qty{\Lambda_e}_{e\in \mathbb{E}}$ be the collection of theories given by \cref{tec} and \cref{hsp}, since each $\Lambda_e\in \mathbf{Form}\qty(\mathcal{L}_e)$ we can define $\Lambda$ as $\bigcup_{e\in \mathbb{E}}\mathbf{I}^e_*\qty(\Lambda_e)$.
	 	We have to show that $\mathcal{A}\in \catname{V}$ if and only if $\mathcal{A}\in \mode{\Lambda}$.
	 	
	 	\begin{itemize}
	 		\item[$\Rightarrow$] Let $\mathsf{Form}\qty(\mathbf{I}^e)\qty(\Gamma_{X_e}) \vdash \mathsf{Form}\qty(\mathbf{I}^e)\qty(\psi) $ be a sequent in $\Lambda$ and let $\iota:\kappa \rightarrow A$ an assignment such that $\mathcal{A}\vDash_\iota\mathsf{Form}\qty(\mathbf{I}^e)\qty(\Gamma_{X_e})$. By point $3$ of  \cref{mrp} this implies $\mathcal{A}\vDash_{\iota \circ \bar{i}_e} \Gamma_{X_e}$, therefore $\mathcal{A}\vDash_{\iota \circ \bar{i}_e} \psi$ and we conclude applying lemma \ref{mrp} again. 
	 		\item[$\Leftarrow$] If $\mathscr{U}_\Sigma(\mathcal{A})=(\emptyset, \rotatebox[origin=c]{180}{!}_{H})$, ($\rotatebox[origin=c]{180}{!}_H$ being the empty map $\emptyset\rightarrow H$) then there are no assignment $\kappa \rightarrow A$ and so $\mathcal{A}$ is in $\mode{\Lambda}$. In the other cases let $\Gamma_{X_e}\vdash \psi$ be in $\Lambda_e$ and $\iota:X_e\rightarrow A$ such that $\mathcal{A}\vDash_{\iota} \Gamma_{X_e}$, since $A\neq \emptyset$ there exists $\hat{\iota}:\kappa \rightarrow A$ such that $\hat{\iota}\circ \bar{i}_e= \iota$ as in the previous point \cref{mrp} implies $\mathcal{A}\vDash_{\hat{\iota}}\mathsf{Form}\qty(\mathbf{I}^e)\qty(\Gamma_{X_e})$, so $\mathcal{A}\vDash_{\hat{\iota}}\mathsf{Form}\qty(\mathbf{I}^e)\qty(\psi)$ and again this is equivalent to $\mathcal{A}\vDash_{\iota}\psi$.\qedhere 
	 	\end{itemize} 
	 \end{proof} 
	 \begin{corollary}
	 	$\catname{V}$ is closed under epimorphisms, (small) products and strong monomorphisms if and only if there exists a language $\mathcal{L}$ and an unconditional theory $\Lambda\in \thr{L}$  such that $\catname{V}=\mode{\Lambda}$.
	 \end{corollary}
	 
\section{Conclusions and future work}\label{sec:conc}
In this paper we have introduced a \emph{fuzzy sequent calculus} to capture equational aspects of fuzzy sets.  While equalities are captured by usual equations, information contained in the membership function is captured by \emph{membership proposition} of the form $\ex{l}{t}$, to be interpreted as ``the membership degree of $t$ is at least $l$''.
We have used a natural concept of \emph{fuzzy algebras} to provide a sound and complete semantics for such calculus in the sense that a formula is satisfied by all the models of a given theory if and only if it is derivable from it using the rules of our calculus. 

As in the classical and quantitative contexts, there is a notion of \emph{free model} of a theory $\Lambda$ and thus an associated monad $\term{\Lambda}$ on the category $\fuz{H}$. However, in general Eilenberg-Moore algebras for such monad are not equivalent to models of $\Lambda$, but we have shown that this equivalence holds if $\Lambda$ is \emph{basic}.
In this direction it would be interesting to better understand the categorical status of our approach, investigating possible links between our notion of fuzzy theory and $\fuz{H}$-Lawvere theories as introduced in full generality by Nishizawa and Power in \cite{nishizawa2009lawvere}. A difference between the two approaches is that for us arities are simply finite sets, while following \cite{nishizawa2009lawvere} a $\fuz{H}$-Lawvere theory arities would be given by finite fuzzy sets.

Finally, using the results provided in \cite{milius2019equational} we have proved that, given a signature $\Sigma$, subcategories of $\alg{\Sigma}$ which are closed under products, strong monomorphisms and epimorphic images correspond precisely with categories of models for \emph{unconditional theories}, i.e. theories axiomatised by sequents without premises. Moreover, using the same results, we have also proved that the categories of models of \emph{theories of type $\mathsf{E}$}, i.e. those whose axioms' premises contain only membership propositions involving variables, are exactly those subcategories closed under products, strong monomorphisms and split epimorphisms.

Our category $\fuz{H}$ of fuzzy sets has crisp arrows and crisp equality: arrows are ordinary functions between the underlying sets and equalities can be judged to be either true or false. A way to further ``fuzzifying'' concepts is to use the topos of \emph{$H$-sets}  over the frame $H$ introduced in \cite{fourman1979sheaves}: this is equivalent to the topos of sheaves over $H$ and contains $\fuz{H}$ as a (non full) subcategory. By construction, equalities and functions are ``fuzzy''. It would be interesting to study an application of our approach to this context. A promising  feature is that in an $H$-set the membership degree function is built-in as simply the equality relation, so it would not be necessary to distinguish between equations and membership propositions.
Even more generally, we can replace $H$ with an arbitrary quantale $\mathcal{V}$ and consider the category of sets endowed with a ``$\mathcal{V}$-valued equivalence relation'' \cite{bkp:concur18}.


\bibliography{bibliog}
\newpage 
	\appendix

	\section{Derivations for \cref{equ}}\label{sec:der}
	\begin{turn}{90}
		\begin{minipage}{550pt}
			\hspace{1pt}\\
			\hspace{1pt}\\
			\hspace{1pt}\\
			\hspace{1pt}\\
			\begin{equation*}
			\inferrule*[right=Fun]{\vdash e \equiv x\cdot x^{-1}\\ 
				\inferrule*[right=Exp]{
					\inferrule*[right=A]{\hspace{1pt}}{\ex{l} {x}\vdash \ex{l}{x}}\\ 
					\inferrule*[right=Exp]{
						\inferrule*[right=A]{\hspace{1pt}}{\ex{l}{x}\vdash \ex{l}{x} }
					}{\ex{l}{x}\vdash \ex{l}{x^{-1}} }}{\ex{l}{x}\vdash \ex{l}{x\cdot x^{-1}}}	
			}{\ex{l}{x}  \vdash \ex{l}{e} }\end{equation*}
			\\ \hspace{1pt} \\ \hspace{1pt} \\ \hspace{1pt} \\ 	\\ \hspace{1pt} \\ \hspace{1pt} \\ \hspace{1pt} \\ \hspace{1pt} \\\hspace{1pt} \\ \hspace{1pt} 
			\begin{equation*}
			\inferrule*[right=Fun]{\vdash y\equiv \qty(y^{-1})^{-1}\\\inferrule*[right=Sub]{
					\inferrule*[right=Fun]
					{\vdash x\equiv y \cdot\qty (\qty(y^{-1} \cdot (x\cdot y)) \cdot y^{-1})\\
						\inferrule*[right=Sub]{\ex{l}{x}	\vdash \ex{l}{y\cdot (x \cdot y^{-1})} }{\ex{l}{y^{-1} \cdot (x\cdot y)}\vdash \ex{l}
							{y \cdot \qty(\qty(y^{-1} \cdot (x\cdot y)) \cdot y^{-1})}}} 		{\ex{l}{ y^{-1} \cdot (x\cdot y)}\vdash \ex{l}{x}}}{\ex{l}{\qty(y^{-1})^{-1}\cdot \qty(x \cdot y^{-1})}\vdash \ex{l}{x}}}{\ex{l}{y\cdot\qty(x\cdot y^{-1})}\vdash \ex{l}{x}}
			\end{equation*}
		\end{minipage} 
	\end{turn}
\end{document}